\documentclass[11pt]{article} 
\usepackage[letterpaper,margin=1in]{geometry} 

\usepackage[backend=biber]{biblatex}
\usepackage{graphicx, amsmath, amsfonts,amsthm, hyperref,subcaption,environ, tabularx, mathtools}
\usepackage[dvipsnames]{xcolor}
\usepackage{thmtools} 
\usepackage[capitalise]{cleveref}
\usepackage[shortlabels]{enumitem}
\usepackage{multirow, multicol}
\usepackage{makecell} 
\usepackage{verbatim}
\usepackage{comment}

\usepackage{authblk} 

\usepackage{tikz}
\usetikzlibrary{arrows,arrows.meta,automata, positioning, calc, intersections, shapes, fit, math, graphs, quotes, decorations.pathreplacing}
\usepackage{circledsteps}

\usepackage{ifthen} 

\newcounter{darrow}

\newcommand{\C}{\mathcal{C}} 

\newcommand{\VCSP}{\textsc{VCSP}}

\newcommand*\circled[1]{\tikz[baseline=(char.base)]{
    \node[shape=circle,draw,inner sep=2pt] (char) {#1};}}
\newcommand*\circledT[1]{\tikz[baseline=(char.base)]{
    \node[shape=circle,draw,inner sep=1pt] (char) {#1};}}

\theoremstyle{plain} 
\newtheorem{theorem}{Theorem} 
\newtheorem{lemma}[theorem]{Lemma} 
\newtheorem{proposition}[theorem]{Proposition}
\newtheorem{conjecture}[theorem]{Conjecture}

\theoremstyle{definition} 
\newtheorem{definition}[theorem]{Definition}
\newtheorem{example}[theorem]{Example}

\theoremstyle{remark} 

\title{Binary constraints on one additional variable\\ can create exponential ascents
for local search}

\author[1]{David A. Cohen}
\author[2]{Peter G. Jeavons}
\author[3,4]{Artem Kaznatcheev}
\author[3,4]{Melle {van Marle}}
\author[3,4]{Sofia {Vazquez~Alferez}}

\affil[1]{Department of Computer Science, Royal Holloway University of London
}
\affil[2]{Department of Computer Science, University of Oxford}
\affil[3]{Department of Mathematics, Utrecht University}
\affil[4]{Department of Information and Computing Sciences, Utrecht University}

\date{}
\begin{document}

\maketitle
\begin{abstract}
Local search in combinatorial optimisation can be viewed as an uphill climb on a corresponding fitness landscape,
where the assignments visited by a strict local search follow an ascent.
This hill-climbing is sometimes surprisingly efficient, but not always.
Since fitness landscapes can be succinctly represented by Boolean valued constraint satisfaction problems (VCSPs), we ask: 
what properties of VCSPs ensure that all ascents are polynomial?
Or alternatively, what are the ``simplest'' VCSPs with exponential ascents?
Prior examples of VCSPs with exponential ascents were built up as a chain of gadgets of constraints.
Here we investigate what happens for simpler starlike graphs.
To establish lower bounds on the \emph{longest} ascent, we construct:
(1) a star of binary constraints with a quadratic longest ascent,
(2) a binary VCSP of treedepth $3$ on $4n + 1$ Boolean variables (made by gluing $2n$ triangles of constraints at a common centre variable) with an exponential longest ascent of length $10\cdot2^n - 9$, and
(3) starlike VCSP of logarithmic vertex cover number with an exponential longest ascent.
For \emph{steepest} ascent, we prove (4) a tight bound of $2(n - 1)$ for VCSPs with star constraint graphs, and construct (5) a starlike tree VCSP with treedepth $3$ with a quadratic steepest ascent.
Finally, we prove upper bounds on the \emph{shortest} ascent from any initial assignment:
(6) linear for VCSPs of treedepth 3; and
(7) $O(2^k (n - k))$ for VCSPs of vertex cover number $k$.
Together (2) and (6) establish treedepth $3$, and (3) and (7) establish logarithmic vertex cover number, as the first structural graph parameters for which the longest and shortest ascents can be exponentially separated.
We discuss the consequences of our results for the parameterized complexity of local search.
\end{abstract}

\section{Introduction}

Local search heuristics are popular methods for combinatorial optimisation~\cite{LocalSearchCO_Thesis,LocalSearch_Book1,LocalSearch_Book2,LocalSearch_Book3,LocalSearch_Book4},
and ascents in fitness landscapes are proving to be a useful lens to study the efficiency or inefficiency of these methods~\cite{hakenSteepest,MaxCutDeg4,repCP,tw7,effective-and-efficient}.
Given a finite set of Boolean variables $V$, a space of \emph{assignments} $\{0,1\}^V$, and an objective function $f: \{0,1\}^V \rightarrow \mathbb{Z}$
(that we call the \emph{fitness function}) to be maximized,
local search starts with some initial assignment $x^0 \in \{0,1\}^V$ and follows some update rule to proceed to another assignment until it eventually finds a good one.
We say that two assignments $x,y \in \{0,1\}^V$ are \emph{adjacent} if they differ on a single variable 
(i.e., $\exists u \in V$ such that $x_u \neq y_u$ and $\forall v \in V \backslash \{u\}, \; x_v = y_v$) 
and $x$ is a local peak if for all adjacent $y$, $f(x) \geq f(y)$.
Based on this,
we say the method is a \emph{strict local search} if the update rule always returns an adjacent assignment of strictly better objective value,
or terminates at a local peak if a better adjacent assignment is not available.
The sequence of such assignments $x^0,x^1,\ldots,x^T$ is an \emph{ascent}: each $x^{t + 1}$ is adjacent to $x^t$, $f(x^{t + 1}) > f(x^t)$, and $x^T$ is a local peak.
A \emph{fitness landscape} is the pair consisting of the fitness function and the rule for when two assignments are adjacent~\cite{W32,evoPLS,repCP}.

For researchers wanting to understand the efficiency or inefficiency of local search, this has set clear research agendas. 
On the one hand, identify families of fitness landscapes that ensure short ascents to their peaks~\cite{MaxCutDeg3Easy,repCP,effective-and-efficient}. 
On the other hand, construct ``simple'' families of fitness landscapes that have exponentially long ascents~\cite{hakenSteepest,MaxCutDeg4,slow-greed,VCSPpw3_allExp}.
The latter approach was at work as early as the 1970s, with \textcite{KM72} constructing a $d$-cube where an ascent taken by a simplex pivot rule visits all $2^d$ assignments on the way to the only peak.
This construction recursively took a $(d-1)$-cube with a long ascent, 
made a ``reversed'' version of it, 
and then added the $d$th variable such that fixing $x_d = 0$ induces the ``original'' $(d-1)$-cube and fixing $x_d = 1$ induces the ``reversed'' cube.
In this way, \textcite{KM72} doubled the long ascent with each new variable.
This Klee-Minty recursive construction has served as a model for many later constructions squeezing exponential ascents into fitness landscapes~\cite{fittestHard2,fittestHard3,LP_DeformedProduct,evoPLS,randomFitter1,randomFitter2,randomFacetBound}.
But given the recursive nature of these constructions, are they really that ``simple''?

\subsection{Representing fitness landscapes by valued constraint graphs}

To measure ``simplicity'', we need a compact representation of fitness landscapes.
To arrive at such a compact representation, notice that any pseudo-Boolean fitness function $f: \{0,1\} \rightarrow \mathbb{Z}$ can be written down as a polynomial:
\begin{equation}
f(x) = \sum_{S \in \mathcal{S}} w(S) \prod_{u \in S} x_u
\label{eq:VCSP_QBF}
\end{equation}
where $\mathcal{S} \subseteq 2^V$ is a \emph{set of scopes},
$w: \mathcal{S} \rightarrow \mathbb{Z} \setminus \{0\}$ is the \emph{constraint weight function},
and each \emph{scope} $S \in \mathcal{S}$ is a set of variables $x_u$ with $u \in S$ that are involved in an AND-constraint of weight $w(S)$.
In this representation, (locally) maximizing $f$ is known as a Boolean \emph{valued constraint satisfaction problem (VCSP)}
with a specific $(V,\mathcal{S},w)$ as the VCSP-instance.\footnote{
Already for binary Boolean VCSPs, finding a global peak is NP-hard and even finding a local peak -- the focus of our work -- is complete for the complexity class of polynomial local search~\cite{PLS_Survey,PLS,W2SAT_PLS1,W2SAT_PLS2}.
This mean that we do not expect even \emph{non-local} algorithms to find local peaks in binary Boolean VCSPs in polynomial time.
But even if non-local algorithms can find local peaks on subclasses of binary Boolean VCSPs, it does not mean that local search algorithms will be able to as well~\cite{tw7,effective-and-efficient,VCSPpw3_allExp}.
}

Elsewhere in the AI literature, the valued constraints that comprise a Boolean VCSP-instance $(V,\mathcal{S},\{C_S\;|\;S \in \mathcal{S}\})$ are often defined as arbitrary pseudo-Boolean functions $C_S: \{0,1\}^S \rightarrow \mathbb{Z}$, instead of our specific $C_S(x[S]) = w(S)\prod_{u\in S}x_u$.
But these two perspectives are equivalent: one can always express an arbitrary valued constraint $C_S$ as a sum of AND-constraints. 
For example, given a binary Boolean valued constraint $C_{\{1,2\}}$ on $x_1,x_2$ that maps $00$ to weight $a$, $10$ to weight $b$, $01$ to weight $c$, and $11$ to weight $d$, we can write 
\begin{align}
C_{\{1,2\}}(x_1,x_2) & = a(1 - x_1)(1 - x_2) + bx_1(1 - x_2) + c(1 - x_1)x_2 +  dx_1x_2 \\
& = a + (b - a)x_1 + (c - a)x_2 + (a + d - b - c)x_1x_2
\end{align} 
which is expressed in our notation as $w(\emptyset)=a$, $w(\{1\})=(b - a)$, $w(\{2\})=(c - a)$, and $w(\{1,2\})=(a - b - c + d)$.  
For Boolean VCSPs of bounded arity, the above can be repeated for each constraint to transform the VCSP-instance into our form in linear time (for the binary case specifically, see Theorem 3.4 in \cite{repCP}).
If all the scopes are at most binary (i.e., $\forall S \in \mathcal{S}, |S| \leq 2$) then we have a \emph{binary Boolean VCSP}.
We take the number of variables, $|V|$, as the size of the binary Boolean VCSP-instance.
We allow binary Boolean VCSP-instances to be restricted to induced subproblems:
\begin{definition}[Induced subproblem]
Given a binary Boolean VCSP-instance $\mathcal{C} = (V,\mathcal{S},w)$, any subset of variables $V' \subset V$, and any background assignment $y \in \{0,1\}^{V \setminus V'}$, let $\mathcal{S}' = \{S \cap V' | S \in \mathcal{S}\}$ be the set of scopes contained in $V'$, and define the weight function $w':\mathcal S'\to \mathbb{Z}\setminus \{0\}$ as follows:

\begin{samepage}
\begin{itemize}
    \item for $\{u,v\}\in \mathcal{S}'$: \quad$w'(\{u,v\}) =  w(\{u,v\})$,
    \item for $\{u\}\in \mathcal{S}'$ : \quad
    $w'(\{u\})  =  w(\{u\}) \ + \ \sum_{v \in V \setminus V' 
    \text{ s.t } \{u,v\} \in \mathcal{S}} w(\{u,v\}) y_v$,
    \item for $\emptyset \in \mathcal{S}'$: \quad $w'(\emptyset) = w(\emptyset) + \sum_{u \in V \setminus V' \text{ s.t } \{u\} \in \mathcal{S}} w(\{u\}) y_u + \sum_{u,v \in V \setminus V' \text{ s.t } \{u,v\} \in \mathcal{S}} w(\{u,v\}) y_u y_v$.~\footnote{
    To maintain non-zero range of the constraint weight function $w'$, we remove 
    $S$ from $\mathcal{S}'$ if $w'(S) = 0$.}
\end{itemize}
then $\mathcal{C}' = (V',\mathcal{S'},w')$ is the \emph{induced subproblem} of $\mathcal{C}$ on $V'$ in background $y$.
\end{samepage}
\label{def:effective-unary} 
\end{definition}

By choosing appropriate forms of valued constraints, binary Boolean VCSPs can easily express other previously studied problems like Hopfield nets~\cite{hakenSteepest}, Max-Cut~\cite{MaxCutDeg4,MaxCutSimple}, weighted 2-SAT~\cite{evoPLS}, and many others.
Further, by viewing the VCSP-instance $(V,\mathcal{S},w)$ as a weighted graph (also called a \emph{valued constraint graph}),
we can look at the structure of this graph as capturing the ``simplicity'' of the fitness landscape that it represents.
Thus, we can parameterize the resulting constraint graphs using standard (hyper)graph parameters like 
max-degree~\cite{MaxCutDeg4,MaxCutSimple}, 
pathwidth~\cite{tw7,pw4,slow-greed,pw2MSc,VCSPpw3_allExp}, 
or, as we will explore in this paper, treedepth and vertex cover number~\cite{td_chapter}.

\subsection{Prior constructions for longest, steepest, and shortest ascents}
\label{sec:chaingadgets}


Early recursive constructions like the Klee-Minty cube 
cannot be represented by sparse valued constraint graphs~\cite{tw7,arityBSc}. 
This motivated more recent work to build up VCSP-instances as a chain of iteratively connected gadgets.~\cite{hakenSteepest,MaxCutDeg4,MaxCutSimple,pw2MSc,slow-greed,VCSPpw3_allExp}
Each gadget is a constant-sized set of variables and `internal' constraints involving only those variables.
Some variables in the $k$th gadget are designated as `connectors' and these are the only variables involved in constraints to variables from the $(k + 1)$th gadget.
Repeating this process builds up a chain of $m$ gadgets.
Since both the number of variables and the number of edges is linear in $m$,
this results in sparse constraint graphs with the max-degree and pathwidth bounded by a constant. 

Broadly, there have been three kinds of local search intractability results for sparse VCSPs:
chain constructions where
(1) some ascent is exponential,
(2) the steepest ascent is exponential (i.e., an ascent taken by a greedy local search), or
(3) all ascents from a designated initial assignment are exponential.
Sometimes it is easier to state these in more concise language as the existence of an exponential (1) longest ascent, (2) steepest ascent, or (3) shortest ascent.
These three kinds of intractability results allow us to establish that a 
construction is not efficiently solvable by (1) some local search algorithms, (2) a particular algorithm like greedy local search, or (3) \emph{any} strict local search algorithm.
We give examples of some of these chain constructions in 
\cref{fig:previous-constructions}.
\begin{figure}[t]
    \centering

    \begin{subfigure}[t]{0.48\textwidth}
        \centering
        \includegraphics[width=\textwidth]{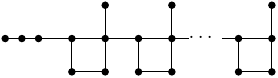}
        \caption{\textcite{MaxCutDeg4}.}
        \label{fig:MaxCutDeg4}
    \end{subfigure} 
    \hfill
    \begin{subfigure}[t]{0.48\textwidth}
        \centering
        \includegraphics[width=\textwidth]{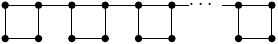}
        \caption{\textcite{repCP}.}
        \label{fig:repCP}
    \end{subfigure}

    \begin{subfigure}[t]{0.48\textwidth}
        \centering
        \includegraphics[width=\textwidth]{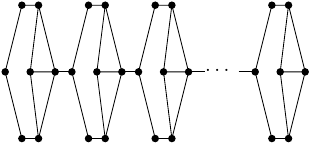}
        \caption{\textcite{hakenSteepest}.}
        \label{fig:hakenSteepest}
    \end{subfigure}
    \hfill
    \begin{subfigure}[t]{0.48\textwidth}
        \centering
        \includegraphics[width=\textwidth]{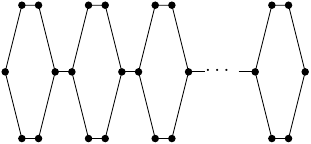}
        \caption{\textcite{slow-greed,pw2MSc}.}
        \label{fig:slow-greed}
    \end{subfigure}
    
    \begin{subfigure}[t]{0.48\textwidth}
        \centering
        \includegraphics[width=\textwidth]{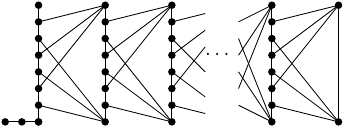}
        \caption{\textcite{MaxCutSimple}.}
        \label{fig:MaxCutSimple}
    \end{subfigure}
    \hfill
    \begin{subfigure}[t]{0.48\textwidth}
        \centering
        \includegraphics[width=\textwidth]{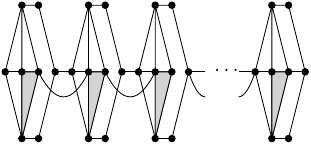}
        \caption{\textcite{VCSPpw3_allExp}.}
        \label{fig:allExp}
    \end{subfigure}
    \caption{Constraint graphs of previous chain constructions of Boolean VCSPs with exponentially long ascents for some ascent \textbf{(a,b)}, steepest ascent \textbf{(c,d)}, or all ascents \textbf{(e,f)}. 
    Constraint weights increase roughly exponentially from right to left in all constructions, with exact weights in the corresponding papers.}
    \label{fig:previous-constructions}
\end{figure}

For results of kind (1), \textcite{MaxCutDeg4} built Max-Cut-instances with an exponential ascent by chaining together gadgets on five variables with the connector variable participating in three internal constraints and one constraint to the next gadget (see \cref{fig:MaxCutDeg4}).
This resulted in a chain construction of a Max-Cut-Instance of maximum degree four.
By generalising from Max-Cut to general valued constraints, \textcite{repCP} showed that the fifth variable in each gadget could be eliminated, reducing the maximum degree of the construction to three while preserving the exponential ascent (see \cref{fig:repCP}).
Both of these constructions had pathwidth two.
This is the lowest possible pathwidth for an exponential ascent from a Boolean VCSP because \textcite{repCP} also showed that for all tree-structured Boolean VCSPs on $n$ variables (i.e., those with treewidth one), the length of any ascent is bounded by $\binom{n+1}{2}$.

For results of kind (2) of the particular ascent followed by greedy local search, \textcite{hakenSteepest} constructed an exponential steepest ascent from a chain of gadgets on seven variables each (see \cref{fig:hakenSteepest}); \textcite{effective-and-efficient} noted that this construction has pathwidth three.
\textcite{slow-greed} and \textcite{pw2MSc} reduced the pathwidth of a construction with an exponential steepest ascent to the lowest possible pathwidth of two (see \cref{fig:slow-greed}).

Finally, for results of kind (3) there has been a series of more complicated chain constructions with the property that all ascents from some designated initial assignment are exponentially long.
\textcite{MaxCutDeg4} claimed that their some-ascent-exponential construction (from \cref{fig:MaxCutDeg4}) can be converted to an all-ascents-exponential construction by growing the gadget from five variables to twenty-eight while preserving a maximum degree of four.
\textcite{MaxCutSimple} provided a simpler chain of Max-Cut gadgets of eight variables with max-degree four and all ascents exponential (see \cref{fig:MaxCutSimple}).
Both of these chain constructions had pathwidth $\geq 4$, but \textcite{VCSPpw3_allExp} 
have constructed a chain of ternary Boolean VCSP gadgets on eight variables that reduces the pathwidth to three.
It is not known if there are Boolean VCSPs of pathwidth two with exponential shortest ascents.

The chain constructions for exponential (1) longest, (2) steepest, and (3) shortest ascents often use gadgets that build on each other.
For example, the exponential \emph{steepest} ascent construction by \cite{slow-greed} and \cite{pw2MSc} is actually built from the chain construction of an exponential 
\emph{longest} ascent in \cite{td_short}.
Similarly, the construction of an exponential shortest ascent in \cite{VCSPpw3_allExp}
started with the six variable gadget for a chain construction with an exponential steepest ascent in \cite{slow-greed} and showed how to make the steepest ascent into the \emph{only} (and thus the shortest) ascent by introducing two ``control'' variables per gadget.
Thus, for the chain construction VCSPs, better gadgets for longest ascent led to better gadgets for steepest ascent, and finally to better gadgets for shortest ascent.
But this also meant that the constraint graphs of these constructions have been of about the same `simplicity' and so no structural graphs parameter has previously been identified for which there is a provable exponential separation between the longest and shortest ascents.

\subsection{Our contribution}

Although the chain construction VCSPs are simple, the fitness landscapes that they represent are built up from sublandscapes that are complicated along every dimension.
No matter how you divide these landscapes into a constant number of induced sublandscapes,
at least one of those sublandscapes will already 
have an exponential ascent.
Is this an essential feature of exponential ascents?
We show here that
this recursive complexity is not an essential feature:
it is possible to take as few as two simple sublandscapes and combine them into a complex one.
To show this, we switch from the chain constructions of prior work to a new kind of star construction.
Here we create the following constructions to lower bound the length of the longest and steepest ascents:\footnote{
Several instances of all five of these constructions have been simulated using Python software~\cite{CohenVCSPSimulator} to verify the existence of the longest and steepest ascents described.
}
\begin{description}
\item[\cref{ex:star_vcsp_quad_ascent}:] Star with a quadratic longest ascent;
\item[\cref{ex:star_of_gadgets}:] Star of gadgets of treedepth $3$ with an exponential longest ascent;
\item[\cref{ex:longest_star_graph_sa}:] Star with $2(n - 1)$ steepest ascent;
\item[\cref{ex:long_star}:] Starlike tree of treedepth $3$ with quadratic steepest ascent. 
\item[\cref{ex:vcn_lowerbound}:] Starlike hypergraph of vertex cover number $\log n$ with an exponential longest ascent.
\end{description}

\noindent In addition to the above lower bounds, we prove upper bounds for steepest and shortest ascents:
\begin{description}[resume]
\item[\cref{prop:star_sa_is_linear}:] Steepest ascents in stars have length of at most $2(n - 1)$;
\item[\cref{thm:short_trees}:] Shortest ascents in trees flip each variable at most once (so have length at most $n$);
\item[\cref{thm:short_td3}:] Shortest ascents in constraint graphs of treedepth $\leq 3$ have length at most $3n - 1$;
\item[\cref{thm:vcn_shortest_upperbound}:] Shortest ascents in constraint graphs of vertex cover number $\leq k$ have length at most $2^k(n - k)$.
\end{description}


\begin{table}
\centering
\begin{tabular}{l||c|c|c}
Valued Constraint Graph Class & Longest Ascent & Steepest Ascent & Shortest Ascent \\
\hline
unary & $n$ & $n$ & $n$  \\
star & $\Theta(n^2)$ & $2(n - 1)$ & $n$  \\
tree & $\binom{n + 1}{2}$ & $\Theta(n^2)$ & $n$  \\
treedepth $\leq 3$ & $2^{\Theta(n)}$ & $\Omega(n^2)$ & $\leq 3n - 1$  \\
vertex cover \# $\leq k$ & ? & ? & $\leq 2^{k}(n - k)$  \\
vertex cover \# $\leq \log n$ & $2^{\Theta(n)}$ & ? & $n^{O(1)}$\\
pathwidth $\leq 2$ & $2^{\Theta(n)}$ & $2^{\Theta(n)}$ & ?  \\
pathwidth $\leq 3$ & $2^{\Theta(n)}$ & $2^{\Theta(n)}$ & $2^{\Theta(n)}$
\end{tabular}
\caption{The worst-case length of the longest ascent, steepest ascent, or shortest ascent in fitness landscapes represented by Boolean VCSPs on $n$ variables with various classes of constraint graphs.}
\label{tab:PC_BooleanVCSP}
\end{table}

In \cref{sec:PCLS}, we discuss the consequences of this simplicity for the parameterized complexity of local search.
We show that our star of gadgets construction has treedepth three (\cref{prop:td3}) -- whereas all the prior chain constructions had treedepth $\Omega(\log n)$.
From this, we conclude that our star of gadgets construction is simpler than the prior state-of-the-art for the existance of long ascents.
We also show that unlike for chain gadget constructions, star gadget constructions cannot easily be modified from having one long ascent into having all ascents long.
Specifically, we show that the shortest ascent for Boolean VCSPs of treedepth three is linear (\cref{thm:short_td3}), giving the first exponential separation between longest and shortest ascent for a fixed graph parameter.
We conjecture that VCSPs of any constant treedepth have polynomial shortest ascents (\cref{conj:treedepth}).
We build evidence for our conjecture by showing that for VCSPs that have the more fine-grained parameter of vertex cover number of $k$, the shortest ascent is bounded by $2^k(n - k)$ (\cref{thm:vcn_shortest_upperbound}).
And provide a second exponential separation between longest and shortest ascent by providing a VCSP with logarithmic vertex cover number and an exponential longest ascent (\cref{ex:vcn_lowerbound}).
These exponential gaps between the shortest and longest ascents transforms the general question: ``is finding local peaks tractable for local search?''
Instead it becomes the question: ``will specific (classes of) local search algorithms end up following one of the short ascents or one of the long ones?''
In \cref{tab:PC_BooleanVCSP}, we summarize this burgeoning field of parametrized complexity of local search.

\section{Longest ascent: stars and windmills}
\label{sec:thickchain}

\textcite{repCP} proved a quadratic upper bound on the length of the longest ascent in any
fitness landscape represented by a tree of binary Boolean valued constraints.
Since stars are trees, this upper bound applies to stars.
We now give a matching lower bound for stars. 

\begin{figure}[!htb]
    \centering
    \begin{tikzpicture}
    \tikzset{state/.style={circle,fill=white,draw=black,text=black, style={minimum size=32pt}}}
    \tikzset{edgeQ/.style={midway, above, sloped}}
    \tikzset{edgeR/.style={near end, above, sloped}}
    \tikzset{edgeS/.style={near end, below, sloped}}
    
    \node[state, label=above:\(4m+1\)](X0){\tiny\(2m + 1\)};
    
    \node[state, label=above right:\(1\)](X1)[right=2cm of X0, yshift=4cm]{\tiny\(1\)};
    \node[state, label=above right:\(1\)](X2)[below=1cm of X1]{\tiny\(2\)};
    \node[below=1.5cm of X2]{\(\vdots\)};
    \node[state, label=above right:\(1\)](X2m)[below=3.5cm of X2]{\tiny\(2m\)};
    
    \node[state, label=above left:\(4m+3\)](X2m+2)[left=3cm of X0, yshift=5.0cm]{\tiny\(2m+2\)};
    \node[state, label=above left:\(4m+3\)](X2m+4)[below=0.7cm of X2m+2]{\tiny\(2m+4\)};
    \node[below=0.1cm of X2m+4]{\(\vdots\)};
    \node[state, label=above left:\(4m+3\)](X4m)[below=1.0cm of X2m+4]{\tiny\(4m\)};
    
    \node[state, label=below left:\(1\)](X2m+3)[left=3cm of X0, yshift=-0.5cm]{\tiny\(2m+3\)};
    \node[state, label=below left:\(1\)](X2m+5)[below=0.7cm of X2m+3]{\tiny\(2m+5\)};
    \node[below=0.1cm of X2m+5]{\(\vdots\)};
    \node[state, label=below left:\(1\)](Xm+k+l)[below=1.0cm of X2m+5]{\tiny\(4m+1\)};

     \draw [decorate,decoration={brace,amplitude=5pt,raise=4ex}]
    (4.5,4.5) -- (4.5,-3.0) node[midway, xshift=1.5cm]{\normalsize\(Q\)};
    \draw [decorate,decoration={brace,amplitude=5pt, mirror, raise=4ex}]
    (-6,5.5) -- (-6,0.5) node[midway, xshift=-1.5cm]{\normalsize\(S\)};
    \draw [decorate,decoration={brace,amplitude=5pt,mirror, raise=4ex}]
    (-6,0) -- (-6,-5.0) node[midway, xshift=-1.5cm]{\normalsize\(R\)};
    
    \path 
    (X0) edge node[edgeQ]  
      	{\(-2\)} (X1)
    (X0) edge node[edgeQ]  
      	{\(-2\)} (X2)
    (X0) edge node[edgeQ] 
      	{\(-2\)} (X2m)
       
    (X0) edge node[edgeR]
      	{\(-(4m+2)\)} (X2m+2)
    (X0) edge node[edgeR]
      	{\(-(4m+2)\)} (X2m+4)
    (X0) edge node[edgeR]
      	{\(-(4m+2)\)} (X4m)
       
    (X0) edge node[edgeS]
      	{\(4m+2\)} (X2m+3)
    (X0) edge node[edgeS]
      	{\(4m+2\)} (X2m+5)
    (X0) edge node[edgeS]
      	{\(4m+2\)} (Xm+k+l)
    ;
    \end{tikzpicture}
    \caption{A star-structured \VCSP-instance
    with an ascent of length $4m^2 + 8m$.
    }
    \label{fig:star_vcsp_quad_ascent}
\end{figure}
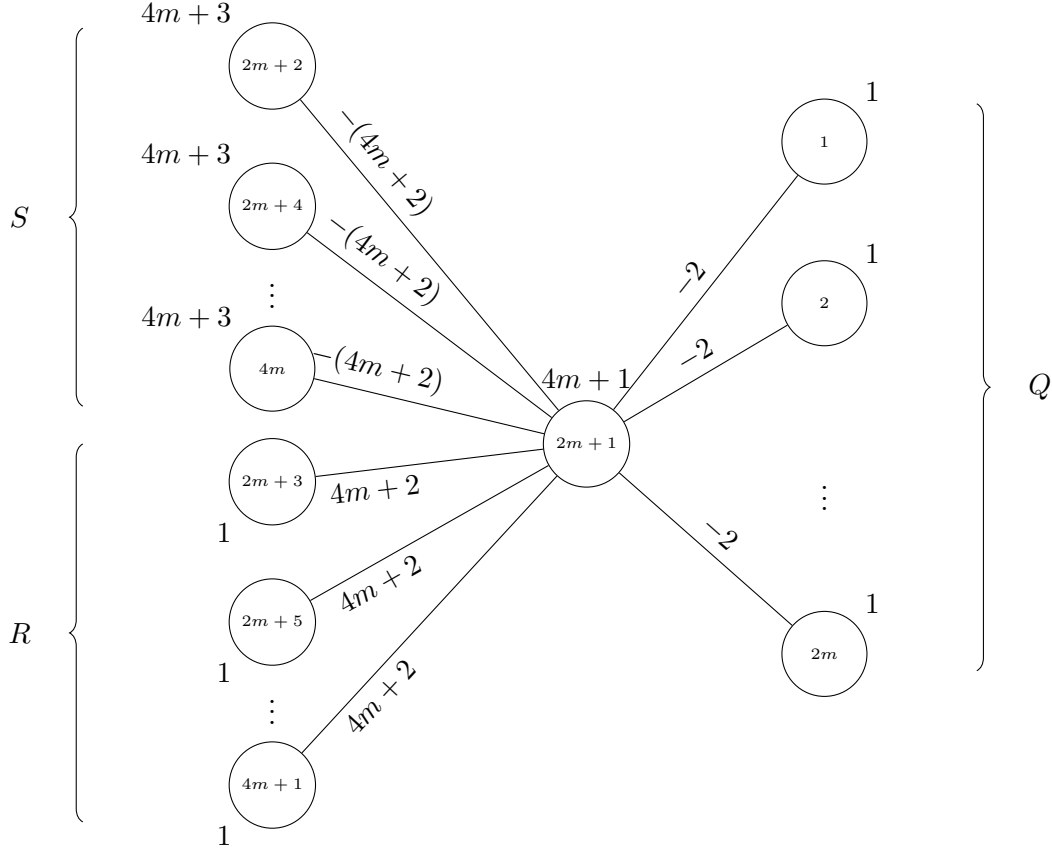

\begin{example}\label{ex:star_vcsp_quad_ascent}
Consider a star on $4m + 1$ variables.
We divide these variables into four sets, 
a set $Q = \{1, \ldots, 2m\}$ of $2m$ variables, one centre variable $c = 2m + 1$, 
and two sets $S = \{2m + 2, 2m + 4, \ldots, 4m - 2, 4m\}$ 
and $R = \{2m + 3, 2m + 5, \ldots, 4m - 1, 4m + 1\}$ of $m$ variables each.
We set the unary and binary constraints as:
\begin{equation*}
    w(\{u\}) = \begin{cases}
        4m + 1 & \text{if } u = c \\
        1 & \text{if } u \in Q \\
        4m + 3 & \text{if } u \in S \\
        1 & \text{if } u \in R
    \end{cases} \quad\quad\text{and}\quad\quad
    w(\{c,u\}) = \begin{cases}
    -2 & \text{if } u \in Q \\
    -(4m + 2) & \text{if } u \in S \\
    4m + 2 & \text{if } u \in R
    \end{cases}.
\end{equation*}
Consider an ascent on the fitness landscape represented by this VCSP-instance
starting from $x = 0^{4m + 1}$ and flipping the variable with the smallest index 
that allows an increase in fitness.
This will first flip variables in the following order: (1) all variables in $Q$, (2) $c$ followed by all the variables in $Q$, (3) alternating variables in $S$ and $R$, where each flip of a $u \in S \cup R$ is followed by a flip in $c$ which is then followed by flips in all the variables in $Q$.
The key observation being that in step (3), every flip of a variable in $S \cup R$ results in $2m + 1$ flips in the variables of lower index.
This ascent will increase fitness by $1$ at each step and end at the fitness peak $1^{2m + 1}0^m$ after $4m^2 + 8m$ steps.
\end{example}

\label{sec:construction}

We now show that it is possible to represent a landscape with an exponential ascent using a ``star of gadgets'' construction where the gadgets intersect in a single variable.
Specifically, we will build a ``windmill" graph 
where $2n$ triangles share a single common centre variable.

\begin{example}\label{ex:star_of_gadgets}
Consider the gadget $\mathcal{C}_k$ on the variables $\{c\} \cup V_k$ where $V_k=\{(1,k),(2,k),(3,k),(4,k)\}$ with five unary constraints and six binary constraints,
as shown in \cref{fig:basic-gadget}, with the weights below:
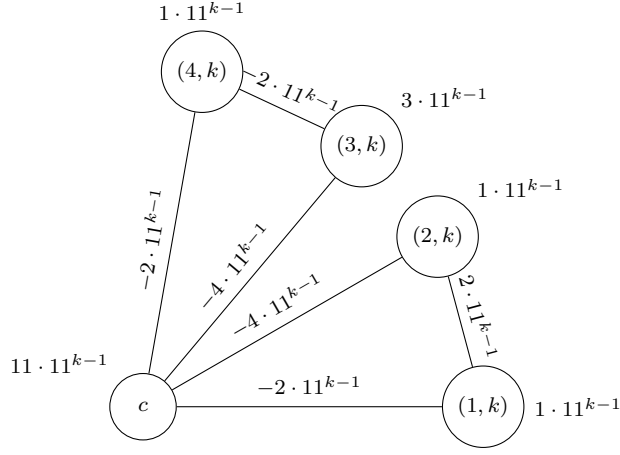
\begin{figure}[h]
    \centering
    \begin{tikzpicture}
    \def\radius{4.5} 

    \node[circle, draw, minimum size=25pt, label=above left:\scriptsize$11\cdot11^{k-1}$] (v-0) at (0:0cm) {\scriptsize$c$};
        
    \node[circle, draw, minimum size=15pt, label=right:\scriptsize$1\cdot 11^{k-1}$] (v-1-1) at (0:\radius cm) {\scriptsize$(1,k)$};
    \node[circle, draw, minimum size=15pt, label=above right:\scriptsize$1\cdot 11^{k-1}$] (v-2-1) at (30:\radius cm) {\scriptsize$(2,k)$};
    \node[circle, draw, minimum size=15pt, label=above right:\scriptsize$3\cdot 11^{k-1}$] (v-3-1) at (50:\radius cm) {\scriptsize$(3,k)$};
    \node[circle, draw, minimum size=15pt, label=above:\scriptsize$1\cdot 11^{k-1}$] (v-4-1) at (80:\radius cm) {\scriptsize$(4,k)$};

    \draw[black] (v-1-1) -- (v-2-1) node [midway, above, sloped] {\scriptsize$2\cdot 11^{k-1}$};
    \draw[black] (v-3-1) -- (v-4-1) node [midway, above, sloped] {\scriptsize$-2\cdot 11^{k-1}$};

    \draw[black] (v-1-1) -- (v-0) node [midway, above, sloped] {\scriptsize$-2\cdot 11^{k-1}$};
    \draw[black] (v-2-1) -- (v-0) node [midway, above, sloped] {\scriptsize$-4\cdot 11^{k-1}$};
    \draw[black] (v-3-1) -- (v-0) node [midway, above, sloped] {\scriptsize$-4\cdot 11^{k-1}$};
    \draw[black] (v-4-1) -- (v-0) node [midway, above, sloped] {\scriptsize$-2\cdot 11^{k-1}$};
\end{tikzpicture}
    \caption{Constraint graph of $\C_k$.
    Unary constraint weights appear near the nodes and binary constraint weights along the edges.}
    \label{fig:basic-gadget}
\end{figure}
\begin{align*}
    w(\{c\})          &=11\cdot11^{k-1} & w(\{(1,k)\})    &=11^{k-1}         & w(\{(2,k)\}) &= 11^{k-1}\\
                   &      & w(\{(3,k)\})    &=3\cdot 11^{k-1}  & w(\{(4,k)\}) &= 11^{k-1}\\
    w(\{c,(1,k)\})&=-2\cdot 11^{k-1} & w(\{c,(2,k)\}) &=-4\cdot 11^{k-1} & w(\{(1,k),(2,k)\})&=2\cdot 11^{k-1} \\
    w(\{c,(3,k)\})&=-4\cdot 11^{k-1} & w(\{c,(4,k)\}) &=-2\cdot 11^{k-1} & w(\{(3,k),(4,k)\})&=-2\cdot 11^{k-1}
\end{align*}
For any positive integer $n$, we construct a VCSP-instance $\C_{\leq n}$ 
on a set of $4n+1$ variables $V=\{c\} \cup \big( \bigcup _{k=1}^nV_k \big)$ as the union of all $\C_k$ for $k=1,2,\ldots,n$,
and set the unary constraint on the shared centre variable $c$ to $w(\{c\}) = 11^n$.
In other words, this common variable keeps only the weight from the final copy, $\C_n$.
The complete construction is illustrated in \cref{fig:triangle-flower}.
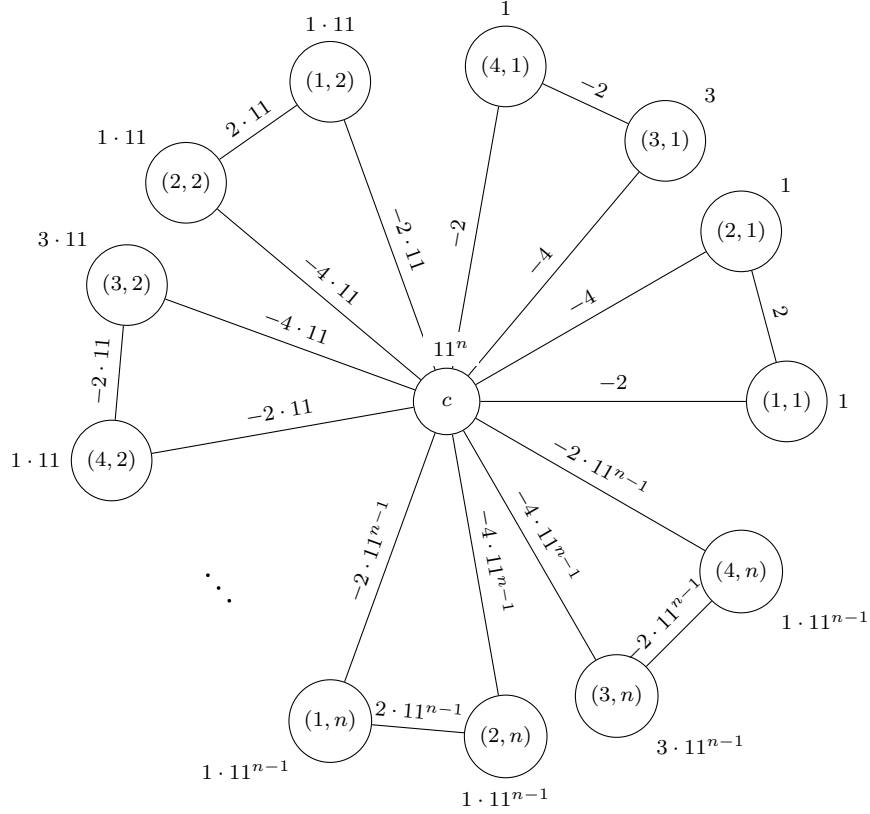
\begin{figure}[htb]
    \centering
    \begin{tikzpicture}
    \def\radius{4.5}

    \node[circle, draw, minimum size=25pt] (v-0) at (0:0cm) {\scriptsize$c$}; 
        
    \node[circle, draw, minimum size=15pt, label=right:\scriptsize$1$] (v-1-1) at (0:\radius cm) {\scriptsize$(1,1)$};
    \node[circle, draw, minimum size=15pt, label=above right:\scriptsize$1$] (v-2-1) at (30:\radius cm) {\scriptsize$(2,1)$};
    \node[circle, draw, minimum size=15pt, label=above right:\scriptsize$3$] (v-3-1) at (50:\radius cm) {\scriptsize$(3,1)$};
    \node[circle, draw, minimum size=15pt, label=above:\scriptsize$1$] (v-4-1) at (80:\radius cm) {\scriptsize$(4,1)$};
    
    \node[circle, draw, minimum size=15pt, label=above:\scriptsize$1\cdot11$] (v-1-2) at (110:\radius cm) {\scriptsize$(1,2)$};
    \node[circle, draw, minimum size=15pt, label=above left:\scriptsize$1\cdot11$] (v-2-2) at (140:\radius cm) {\scriptsize$(2,2)$};
    \node[circle, draw, minimum size=15pt, label=above left:\scriptsize$3\cdot11$] (v-3-2) at (160:\radius cm) {\scriptsize$(3,2)$};
    \node[circle, draw, minimum size=15pt, label=left:\scriptsize$1\cdot11$] (v-4-2) at (190:\radius cm) {\scriptsize$(4,2)$};

    \node[circle, draw, minimum size=15pt,label=below left:\scriptsize$1\cdot11^{n-1}$] (v-1-n) at (250:\radius cm) {\scriptsize$(1,n)$};
    \node[circle, draw, minimum size=15pt,label=below:\scriptsize$1\cdot11^{n-1}$] (v-2-n) at (280:\radius cm) {\scriptsize$(2,n)$};
    \node[circle, draw, minimum size=15pt,label=below right:\scriptsize$3\cdot11^{n-1}$] (v-3-n) at (300:\radius cm) {\scriptsize$(3,n)$};
    \node[circle, draw, minimum size=15pt,label=below right:\scriptsize$1\cdot11^{n-1}$] (v-4-n) at (330:\radius cm) {\scriptsize$(4,n)$};
    \path (v-4-2) -- (v-1-n) node [midway, sloped] {\Large\ldots};

    \draw[black] (v-1-1) -- (v-2-1) node [midway, above, sloped] {\scriptsize$2$};
    \draw[black] (v-3-1) -- (v-4-1) node [midway, above, sloped] {\scriptsize$-2$};

    \draw[black] (v-1-2) -- (v-2-2) node [midway, above, sloped] {\scriptsize$2\cdot 11$};
    \draw[black] (v-3-2) -- (v-4-2) node [midway, above, sloped] {\scriptsize$-2\cdot 11$};

    \draw[black] (v-1-n) -- (v-2-n) node [midway, above, sloped] {\scriptsize$2\cdot 11^{n-1}$};
    \draw[black] (v-3-n) -- (v-4-n) node [midway, above, sloped] {\scriptsize$-2\cdot 11^{n-1}$};

    \draw[black] (v-1-1) -- (v-0) node [midway, above, sloped] {\scriptsize$-2$};
    \draw[black] (v-2-1) -- (v-0) node [midway, above, sloped] {\scriptsize$-4$};
    \draw[black] (v-3-1) -- (v-0) node [midway, above, sloped] {\scriptsize$-4$};
    \draw[black] (v-4-1) -- (v-0) node [midway, above, sloped] {\scriptsize$-2$};

    \draw[black] (v-1-2) -- (v-0) node [midway, above, sloped] {\scriptsize$-2\cdot 11$};
    \draw[black] (v-2-2) -- (v-0) node [midway, above, sloped] {\scriptsize$-4\cdot 11$};
    \draw[black] (v-3-2) -- (v-0) node [midway, above, sloped] {\scriptsize$-4\cdot 11$};
    \draw[black] (v-4-2) -- (v-0) node [midway, above, sloped] {\scriptsize$-2\cdot 11$};

    \draw[black] (v-1-n) -- (v-0) node [midway, above, sloped] {\scriptsize$-2\cdot 11^{n-1}$};
    \draw[black] (v-2-n) -- (v-0) node [midway, above, sloped] {\scriptsize$-4\cdot 11^{n-1}$};
    \draw[black] (v-3-n) -- (v-0) node [midway, above, sloped] {\scriptsize$-4\cdot 11^{n-1}$};
    \draw[black] (v-4-n) -- (v-0) node [midway, above, sloped] {\scriptsize$-2\cdot 11^{n-1}$};

    \node[fill=white] (label) at (85:0.7cm) {\scriptsize$11^n$};
\end{tikzpicture}
    \caption{Constraint graph of $\C_{\leq n}$. 
    Unary constraint weights appear next to the nodes and binary constraint weights along the edges.}
    \label{fig:triangle-flower}
\end{figure}
\end{example}
\noindent The fitness landscape of the basic gadget $\C_{k}$ has an ascent of length 11, 
from assignment $00000$ to assignment $10000$, as shown in Table~\ref{tab:simple-gadget-ascent}.
Note that the centre variable $c$ is flipped twice from $0$ to $1$ at steps \protect\circled{$4$} and \protect\circled{$8$} with variables in $V_k$ set to $y^{\scriptsize\circled{$4$}} = 0111$ and $y^{\scriptsize\circled{$8$}} = 1110$.
Looking at the induced subproblems on just the centre variable $\mathcal{C}'_k = (\{c\},\{\{c\}\},w')$ with background set to $y^{\scriptsize\circled{4}}$ 
and $\mathcal{C}''_k = (\{c\},\{\{c\}\},w'')$  with background set to $y^{\scriptsize\circled{8}}$, the weights are chosen such that the induced unary weight function on the centre variable is given by:
\begin{align}
w'(\{c\}) & = w(\{c\}) + w(\{c,(2,k)\}) + w(\{c,(3,k)\}) + w(\{c,(4,k)\}) \\
    & = (11 - 4 - 4 - 2)\cdot11^{k - 1} = 11^{k - 1} \label{eq:w1}\\
w''(\{cx\}) & = w(\{c\}) + w(\{c,(1,k)\}) + w(\{c,(2,k)\}) + w(\{c,(3,k)\}) \\
    & = (11 - 2 - 4 - 4)\cdot11^{k - 1} = 11^{k - 1} \label{eq:w2}
\end{align}
which is the same weight as $\mathcal{C}_{k - 1}$ (and thus also $\mathcal{C}_{\leq k - 1}$) places on $\{c\}$.
We will use this property at steps \protect\circled{$4$} and \protect\circled{$8$} to recursively prove an exponential ascent in $\mathcal{C}_{\leq n}$:

\begin{table}[htb]
    \centering
    \begin{tabular}{c|ccccc|c}
         Step & ~~$x_c$~~ & $x_{(1,k)}$ & $x_{(2,k)}$ & $x_{(3,k)}$ & $x_{(4,k)}$ & $f(x)/11^{k-1}$ \\
         \hline
            & 0 & 0 & 0 & 0 & 0 & 0 \\
         \protect\circled{$1$}  & 0 & 0 & 0 & 0 & 1 & 1 \\
         \circled{$2$}  & 0 & 0 & 0 & 1 & 1 & 2  \\
         \circled{$3$}  & 0 &0 &1 &1 &1 &3 \\
         \hline
         \circled{$4$}  & 1 &0 &1 &1 &1 &4 \\
         \circled{$5$}  & 1 &1 & 1 &1 &1 &5 \\
         \circled{$6$}  & 0 & 1 & 1 &1 &1 &6 \\
         \circled{$7$}  & 0 &1 &1 &1 &0 &7 \\
         \hline
         \circled{$8$}  & 1 &1 &1 &1 &0 &8 \\
         \circled{$9$}  & 1 &1 &1 &0 &0 &9 \\
         \circledT{$10$}  & 1 &1 &0 &0 &0 &10 \\
         \circledT{$11$}  & 1& 0 &0 &0 &0 &11 \\
    \end{tabular}
    \caption{An 11 step ascent  in the fitness landscape represented by $\C_{k}$.}
    \label{tab:simple-gadget-ascent}
\end{table}

\begin{theorem}\label{thm:exponential-ascent}
    The fitness landscape of $\C_{\leq n}$ from \cref{ex:star_of_gadgets} has an ascent of length $10\cdot2^n-9$. 
\end{theorem}
\begin{proof}
    By induction on $n$, we will show that, for each value of $n \geq 0$, 
    there is an ascent in the fitness landscape of $\C_{\leq n}$ of length $10\cdot2^n-9$ 
    that starts at the assignment $0(0000)^{n}$, and ends at $1(0000)^{n}$.

    \noindent \textbf{Base Case:} 
    For $n=0$ we have $\mathcal{C}_{\leq 0}$ on the centre variable 
    $\{c\}$ with $w(\{c\}) = 1$, so there exists an ascent of length one from $0$ to $1$.
    
    Assume by induction that the result holds for $\C_{\leq k - 1}$,
    and look at an ascent in the landscape represented by $\C_{\leq k}$:

    
    \noindent \textbf{Steps \protect\circled{$1$},\protect\circled{$2$} and \protect\circled{$3$}:} 
    Take the initial assignment $0 (0000)^{k}$, and take the first three steps (\protect\circled{$1$},\protect\circled{$2$} and \protect\circled{$3$}) of the ascent shown in \cref{tab:simple-gadget-ascent} by flipping the corresponding variables of $\C_{k}$
    (i.e., flip variable $(4,k)$ to $1$, then variable $(3,k)$ to $1$, then variable $(2,k)$ to $1$).
    This takes us to the assignment $0(0000)^{k-1}{0111}$.

    \noindent \textbf{Step \protect\circled{$4$}:} 
    Now consider the induced subproblem $\C'_{\leq k-1}$ on the variables $V_{\leq k-1}$ with background $y = {0111}$ on $V_k$. 
    By \cref{def:effective-unary,eq:w1}, the unary constraint on the centre variable $c$ in $\C'_{\leq k-1}$ is equal to $w'(\{c\}) = 11^{k-1}$.
    Hence, $\C'_{\leq k-1} = \C_{\leq k-1}$, so by the inductive hypothesis we can make an ascent of length $10\cdot2^{k-1}-9$ from $0(0000)^{k-1}{0111}$ to $1(0000)^{k-1}{0111}$ by flipping variables in $V_{\leq k-1}$ in some suitable order.  
    These flips have the effect in $\C_{k}$ of performing step \protect\circled{$4$}.
    
    \noindent \textbf{Steps \protect\circled{$5$},\protect\circled{$6$} and \protect\circled{$7$}:} 
    At this point we fix the variables in $V_{\leq k-1}$ and continue with the ascent shown in Table~\ref{tab:simple-gadget-ascent} for the three steps \protect\circled{$5$},\protect\circled{$6$} and \protect\circled{$7$} by flipping the corresponding variables of $\C_{k}$
    (i.e., flip variable $(1,k)$ to $1$, then variable $c$ to $0$, then variable $(4,k)$ to $0$).
    This takes us to the assignment $0(0000)^{k-1}{1110}$.

    \noindent \textbf{Step \protect\circled{$8$}:} Now consider the induced subproblem $\C''_{\leq k-1}$ on the variables of $\C_{\leq k-1}$ with background $y = {1110}$. 
    By \cref{def:effective-unary,eq:w2}, the unary constraint on the centre variable $c$ in $\C''_{\leq k-1}$ is equal to $w''(\{c\}) = 11^{k-1}$.
    Hence, $\C''_{\leq k - 1} = \C_{\leq k - 1}$, so by the inductive hypothesis we can make an ascent of length $10\cdot2^{k-1}-9$ from $0(0000)^{k-1}{1110}$ to $1(0000)^{k-1}{1110}$ by flipping variables of $\C_{\leq k-1}$ in some suitable order.
    These flips have the effect in $\C_{k}$ of performing step \protect\circled{$8$}.
    
    \noindent \textbf{Steps \protect\circled{$9$},\protect\circledT{$10$} and \protect\circledT{$11$}:}
    At this point we fix the variables in $\C_{\leq k-1}$ and continue with the ascent shown in Table~\ref{tab:simple-gadget-ascent} for the three steps \protect\circled{$9$},\protect\circledT{$10$} and \protect\circledT{$11$} by flipping the corresponding variables of $\C_{k}$.  
    (i.e., flip variable $(3,k)$ to $0$, then variable $(2,k)$ to $0$, then variable $(1,k)$ to $0$).
    This completes the 11 step ascent and takes us to the assignment $1(0000)^{k-1}{0000}$.

    The total length of the ascent we have just described is $3 + (10\cdot2^{k-1}-9) + 3 + (10\cdot2^{k-1}-9) + 3$ which equals $10\cdot2^{k}-9$.
    Hence, by induction, the result follows for all values of $n$.
\end{proof}

\section{Steepest ascent: stars and starlike trees}
\label{sec:steepest_ascent}

We will show that the length of a steepest ascent on the landscape represented by any Boolean VCSP instance on \(n\) variables whose constraint graph is a star is at most \(2(n-1)\).
Before proving this fact, we begin by providing a simple example that saturates this bound.

\begin{example}\label{ex:longest_star_graph_sa}
\textbf{(Star-VCSP with linear steepest ascent)}
    We define a VCSP instance over \(n\) Boolean valuables \(1\dots,n\)
    as shown in \cref{fig:steepest_star}.
    We set variable \(1\) to be the centre of the star.
    We set the unary and binary constraint as 
    \begin{equation*}
        w(\{u\}) = 
        \begin{cases}
        1 & \text{if } u \in \{1, n\} \\
        2 & \text{if } u \in \{2, \dots, n-1\}
        \end{cases}\quad\;\text{and}\quad\;
        w(\{1,u\}) = 
        \begin{cases}
        -3 & \text{if } u \in  \{2, \dots, n-1\}\\
        3(n - 2) & \text{if } u = n
        \end{cases}.
    \end{equation*}
    
\begin{figure}[hbt]
\centering
\begin{tikzpicture}[-, >=stealth', shorten >=1pt, auto, node distance=3cm]
    \tikzset{state/.style={circle, draw, minimum size=26pt}}
    \tikzset{weightNode/.style={midway, above, sloped}}
    \def\radius{2.5} 
    \def\angle{22.5} 

    \node[state, label=above: \scriptsize\(1\)] (v0) {\scriptsize$1$};
    \node[state, label=above right:\scriptsize\(2\)] (v1) at (2*\angle:\radius * 1.25cm){\scriptsize$2$};
    \node[state, label=above right:\scriptsize\(2\)] (v2) at (\angle:\radius cm){\scriptsize$3$};
    \node[state,draw=white] (dots) at (0:\radius * 0.95cm){$\vdots$};
    \node[state, label=right: \scriptsize\(2\)] (vK) at (-\angle:\radius cm){\scriptsize${n-1}$};
    \node[state, label=above left: \scriptsize\(1\)] (vK+1) at (180:\radius cm) {\scriptsize${n}$};

    \draw[black] (v0) -- (v1) node[weightNode]{\scriptsize\(-3\)};
    \draw[black] (v0) -- (v2) node[weightNode]{\scriptsize\(-3\)};
    \draw[black] (v0) -- (vK) node[weightNode]{\scriptsize\(-3\)};
    \draw[black] (v0) -- (vK+1) node[weightNode]{\scriptsize\(3(n-2)\)};
    
\end{tikzpicture}
\caption{Star of constraints with a steepest ascent of length $2(n-1)$.}
\label{fig:steepest_star}
\end{figure}
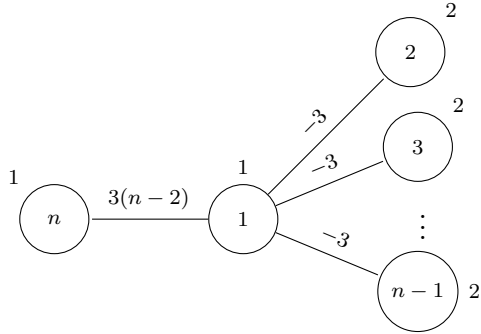
Consider the steepest ascent starting at \(0^{n}\). 
Since the centre variable \(1\) starts at value \(0\), none of the binary constraints come into play. 
Thus, the first step consists of one of the leaves in the variable set \(\{2, \dots, n-1\}\) flipping to \(1\), since these have the largest unary constraint.
In the first \(n-2\) steps, all these leaves flip to \(1\).
Finally, leaf variable \(n\) flips to \(1\).

After all the leaf variables have flipped to \(1\), the induced unary on the centre variable is equal to
\(
    w(\{1\})+\sum_{i=2}^nw(\{1,i\})=1+3(n-2)-3(n-2)=1.
\)
Thus, the centre flips to \(1\) at this point.
After this flip of the centre, each leaf variable in the set $\{2,\ldots,n-1\}$ 
has induced unary $w(\{i\})+w(\{1,i\})=-1$.
Since both the centre variable and variable \(n\) have a positive induced unary and are set to \(1\), all of the variables in \(i\in\{2, \dots, n-1\}\) will now flip back to \(0\), at which point a local maximum has been reached by the steepest ascent in a total of \(2(n-1)\) steps.
\end{example}

The steepest ascent in \cref{ex:longest_star_graph_sa} is the maximum length possible in stars:

\begin{theorem}\label{prop:star_sa_is_linear}
If $\mathcal{C}$ is a Boolean VCSP instance on \(n\) variables whose constraint graph is a star, 
then any steepest ascent on \(\mathcal{C}\) will have length at most \(2(n-1)\).
\end{theorem}

\begin{proof}
Name the centre variable $c$ and the set of leaf variables $V$,
consider any steepest ascent \(x^0, \dots, x^T\) in the fitness landscapes implemented by \(\mathcal{C}\).
If $c$ never flips then each leaf can flip at most once and the length of the ascent is $\leq n - 1$.
Now consider the case where $c$ does flip:
\\
\noindent \textbf{To analyse the first flip of $c$}: 
relabel the domains of $c$ and $V$ so that we can have $x^0 = 0^{n}$. 
Let $t_1$ be the time that $c$ flipped from $0$ to $1$ (i.e., $x^{t_1 - 1}_c = 0$ and $x^{t_1}_c = 1$) and let $V_\text{in}$ be the set of leaves that flipped before $t_1$ with the first flip coming from $u \in V_\text{in}$.
We have that:
\begin{equation}
    w(\{c\}) \ \leq \ w(\{u\}) \ \leq \ \sum_{v \in V_\text{in}} w(\{v\})
    \label{eq:wc_first_upperbound}
\end{equation}
where the first inequality comes from the steepest property and the second inequality comes from $u \in V_\text{in}$ and from $\forall v \in V_\text{in}$ being able to flip with the center at zero (i.e., $w(\{v\}) > 0$).
At time $t_1-1$, we have the induced unary on $c$ as:
\begin{equation}
0 \ < \ w(\{c\}) + \sum_{v \in V_\text{in}}w(\{c,v\}) \ \leq\  \sum_{v \in V_\text{in}} \Bigg( w(\{v\}) + w(\{c,v\}) \Bigg )
\label{eq:Vin_bound}
\end{equation}
where the first inequality comes from the centre being able to flip at step $t_1$ and the second inequality from \cref{eq:wc_first_upperbound}.
For \cref{eq:Vin_bound} to be satisfied, we need at least one $w \in V_\text{in}$ with $w(\{w\})  + w(\{c,w\}) > 0$ and thus not be able to flip down to $0$ after time $t_1$.
Thus, at least one leaf flips at most once and it must flip before the centre does. 
So if $|V_\text{out}|$ is the number of leaf flips after the last flip of the centre then we must have $|V_\text{out}| \leq |V_\text{in}| - 1$.
%
\\
\noindent \textbf{To analyze subsequent flips of $c$}: 
let \(t_1, \dots, t_h\) be the times at which \(c\) flips.
Now consider any $r < h$ and let $Q$ be the subset of leaf variables that have flipped before $t_r$ and also flip between steps $t_r$ and $t_{r + 1}$.
Relabel the domains of all variables such that at time $t_r - 1$, the center variable is zero, all variables in $Q$ are $1$, and the other leaf variables are $0$.
In other words, such that $x^{t_r - 1} = 01^Q 0^{V \setminus Q}$.
Based on this labelling consider a $u \in Q$, since $u$ flips from $1$ to $0$ when $x_c = 1$, we have $w(\{u\}) + w(\{c,u\}) < 0$ and since it previously flipped from $0$ to $1$, it must have done it when $x_c$ was zero, so $0 < w(\{u\}) < -w(\{c,u\})$.
Let $\delta = \min ( w(\{c\}) + \sum_{u \in Q} w(\{c,u\}), \min_{u \in Q } w(\{u\}))$.
We will use the step size of the $t_r$th step of the steepest ascent to lower bound the size the unary on $c$ in this labeling:
\begin{align}
\mathcal{C}(x^{t_r}) - \mathcal{C}(x^{t_r - 1}) & = w(\{c\}) + \sum_{u \in Q} w(\{c,u\})   \geq \delta \\
w(\{c\}) & \geq \delta -\sum_{u \in Q} w(\{c,u\}) > \delta + \sum_{u \in Q} w(\{u\}) \geq (|Q| + 1)\delta
\label{eq:wc_lowerbound}
\end{align}

Let $S$ be a subset of leaf variables such that $u \in S$ if $w(\{c,u\}) < 0$ and $u$ flips from $0$ to $1$ between steps $t_r$ and $t_{r + 1}$.
Since $u \in S$ flipped from $0$ to $1$ while $x_c = 1$, this means that $w(\{u\}) + w(\{c,u\}) > 0$ so
these variables always prefer the $x_u = 1$ assignment independent of $x_c$ and so will only flip once during any ascent.
Since $\delta$ is an upper bound on the smallest step size that the steepest ascent has on the $x_c = 0$ induced sublandscape before step $t_r$ and $u \in S$ did not flip before $t_r$, it means that for all $u \in S$ we have $\delta \geq w(\{u\}) \geq -w(\{c,u\})$.

For all other variables $v \in V \setminus S$ that flipped from $0$ to $1$ between $t_r$ and $t_{r  +1}$, we have $w(\{c,v\}) > 0$.
Now we can use the $t_{r+1}$th step of the steepest ascent to upper bound the size of the unary on $c$:
\begin{align}
0 < \mathcal{C}(x^{t_{r + 1}}) - \mathcal{C}(x^{t_{r + 1} - 1}) & \leq -w(\{c\}) - \sum_{u \in S} w(\{c,u\}) \\
 w(\{c\}) & < -\sum_{u \in S} w(\{c,u\}) \leq \sum_{u \in S} w(\{u\}) \leq |S|\delta
\label{eq:wc_upperbound}
\end{align}

\noindent Combining \cref{eq:wc_lowerbound} and \cref{eq:wc_upperbound}, we get $|Q| + 1 < |S|$.

Since this holds for each $r < h$, we can let $q_r$ be the set of extra flips of leaf variables in block $r$ and $S_r$ be the set of variables that flip for the only time in block $r$.
From this, we have that the total number of flips is:

\begin{align}
    \underbrace{|V_\text{in}| + 1}_\text{first flip of leaves and centre} + \sum_{r = 1}^{h - 1} \underbrace{\Bigg (|S_r| + q_r + 1 \Bigg)}_{\text{flips of leaves and centre in block } r} + \underbrace{|V_\text{out}|}_{\text{flips of leaves after } t_h} \\
    \ \leq \ 2(|V_\text{in}| + \sum_{r = 1}^{h - 1} |S_r| ) 
    \ \leq \ 2(n - 1)
\end{align}
\noindent where the last inequality follows from the total number of leaves being greater or equal than the number of leaves that flipped before the first center flip $V_\text{in}$ and the number of leaves that only ever flip during a block $r$.
\end{proof}

Our proof of the linear bound in \cref{prop:star_sa_is_linear}
depended on the fact that some of the leaf variables
have to have a small weight to prevent them flipping 
too soon in a steepest ascent. 
In \cref{ex:long_star}, we show that by adding 
additional ``unlocking" variables connected to some of the leaves
we can achieve a quadratic steepest ascent similar to the longest ascent in \cref{ex:star_vcsp_quad_ascent}:

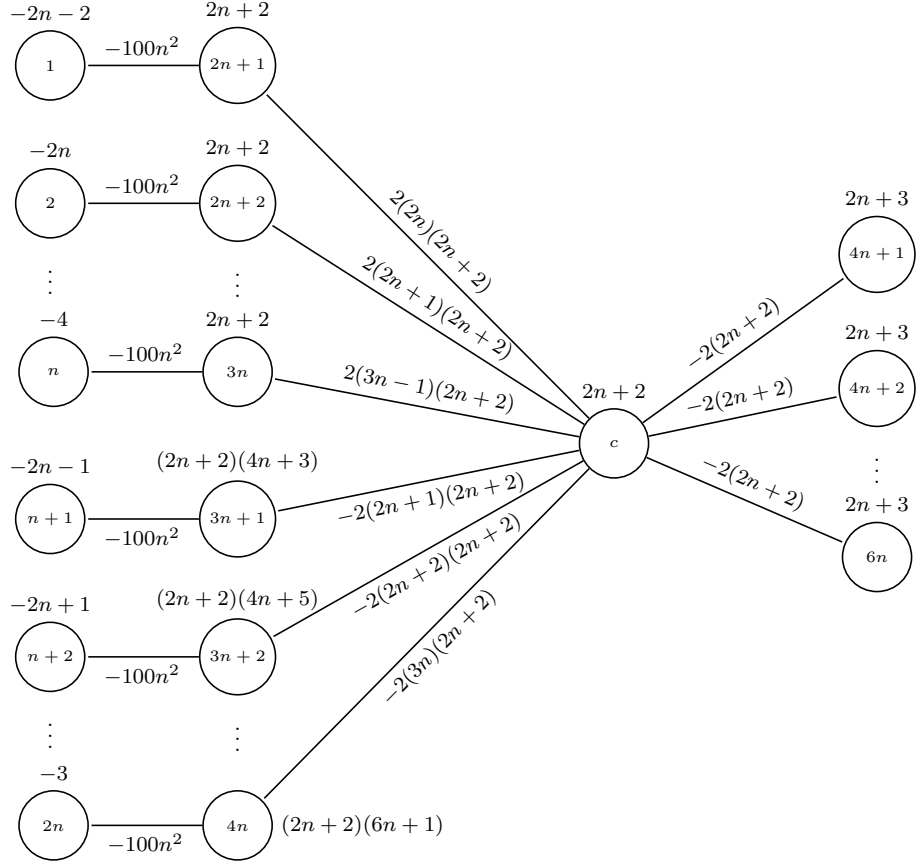
\begin{figure}[!htb]
    \centering
    \scriptsize{
    \begin{tikzpicture}[-,>=stealth',shorten >=1pt,auto,node distance=2.5cm, semithick]
    \tikzstyle{every state}=[fill=white,draw=black,text=black, style={minimum size=26pt}]
    \node[state, label=above:\(2n+2\)](X0){\tiny\(c\)};
    \node[state, label=above:\(2n+3\)](X1)[right=2.5cm of X0, yshift=2.5cm]{\tiny\(4n+1\)};
    \node[state, label=above:\(2n+3\)](X2)[below=0.75cm of X1]{\tiny\(4n+2\)};
    \node[below=0.1cm of X2]{\(\vdots\)};
    \node[state, label=above:\(2n+3\)](X2n)[below=1.25cm of X2]{\tiny\(6n\)};
    
    \node[state, label=above:\(2n+2\)](X1_1)[left=4cm of X0, yshift=5.0cm]{\tiny\(2n+1\)};
    \node[state, label=above:\(2n+2\)](X2_1)[below=0.8cm of X1_1]{\tiny\(2n+2\)};
    \node[below=0.1cm of X2_1]{\(\vdots\)};
    \node[state, label=above:\(2n+2\)](Xn_1)[below=1.25cm of X2_1]{\tiny\(3n\)};
    
    \node[state,label=above:\((2n+2)(4n+3)\)](Xn+1_1)[left=4cm of X0, yshift=-1.0cm]{\tiny\(3n+1\)};
    \node[state,label=above:\((2n+2)(4n+5)\)](Xn+2_1)[below=0.8cm of Xn+1_1]{\tiny\(3n+2\)};
    \node[below=0.1cm of Xn+2_1]{\(\vdots\)};
    \node[state,label=right:\((2n+2)(6n+1)\)](X2n_1)[below=1.25cm of Xn+2_1]{\tiny\(4n\)};
    
    \node[state, label=above:\(-2n - 2\)](X1_2)[left=1.5cm of X1_1]{\tiny\(1\)};
    \node[state, label=above:\(-2n\)](X2_2)[left=1.5cm of X2_1]{\tiny\(2\)};
    \node[below=0.1cm of X2_2]{\(\vdots\)};
    \node[state, label=above:\(-4\)](Xn_2)[left=1.5cm of Xn_1]{\tiny\(n\)};
    
    \node[state,label=above:\(-2n - 1\)](Xn+1_2)[left=1.5cm of Xn+1_1]{\tiny\(n+1\)};
    \node[state,label=above:\(-2n +1\)](Xn+2_2)[left=1.5cm of Xn+2_1]{\tiny\(n+2\)};
    \node[below=0.1cm of Xn+2_2]{\(\vdots\)};
    \node[state,label=above:\(-3\)](X2n_2)[left=1.5cm of X2n_1]{\tiny\(2n\)};

    \tikzset{edgeQ/.style={midway, above, sloped}}
    \tikzset{edgeR/.style={midway, above, sloped}}
    \tikzset{edgeS/.style={midway, below, sloped}}
    
    \path 
    (X0) edge node[edgeQ]  
      	{\(-2(2n+2)\)} (X1)
    (X0) edge node[edgeQ]  
      	{\(-2(2n+2)\)} (X2)
    (X0) edge node[edgeQ] 
      	{\(-2(2n+2)\)} (X2n)
       
    (X0) edge node[edgeR]
      	{\(2(2n)(2n+2)\)} (X1_1)
    (X0) edge node[edgeR]
      	{\(2(2n + 1)(2n+2)\)} (X2_1)
    (X0) edge node[edgeR]
      	{\(2(3n - 1)(2n+2)\)} (Xn_1)
       
    (X0) edge node[edgeS]
      	{\(-2(2n+1)(2n+2)\)} (Xn+1_1)
    (X0) edge node[edgeS]
      	{\(-2(2n+2)(2n+2)\)} (Xn+2_1)
    (X0) edge node[edgeS]
      	{\(-2(3n)(2n+2)\)} (X2n_1)
        
    (X1_1) edge node[midway, above]
      	{\(-100n^2\)} (X1_2)
    (X2_1) edge node[midway, above]
      	{\(-100n^2\)} (X2_2)
    (Xn_1) edge node[midway, above]
      	{\(-100n^2\)} (Xn_2)
        
    (Xn+1_1) edge node[midway, below]
      	{\(-100n^2\)} (Xn+1_2)
    (Xn+2_1) edge node[midway, below]
      	{\(-100n^2\)} (Xn+2_2)
    (X2n_1) edge node[midway, below]
      	{\(-100n^2\)} (X2n_2)
        
    ;  
    \end{tikzpicture}
        }
    \caption{Starlike tree VCSP from Example~\ref{ex:long_star} with a steepest ascent of length \(4n^2 + 8n\). }
    \label{fig:long_star_vcsp}
\end{figure}

\begin{example}\textbf{(Starlike tree with quadratic steepest ascent)}\label{ex:long_star}
We define a VCSP instance  over $6n+1$ variables \(\{c\}\cup \{1, \dots, 6n\}\)
We set the unary constraints as:
\begin{equation*}
    w(\{u\}) = 
    \begin{cases}
        2n+2 & \text{if } u=c\\
        -2n-4+2u & \text{if } u\in\{1, 2, \dots, n\}\\
        -4n-3+2u & \text{if } u\in\{n+1, n+2, \dots, 2n\}\\
        2n+2 & \text{if } u\in\{2n+1, 2n+2, \dots, 3n\}\\
        (2n+2)(2u-2n+1) & \text{if } u\in\{3n+1, 3n+2, \dots, 4n\}\\
        2n+3 & \text{if } u\in\{4n+1, 4n+2, \dots, 6n\}
    \end{cases}
\end{equation*}
\noindent and binary constraints as:
\begin{equation*}
    w(\{v,u\}) = 
    \begin{cases}
        2(u-1)(2n+2) & \text{if } v=c \text{ and } u\in\{2n+1, 2n+2, \dots, 3n\}\\
        -2(u-n)(2n+2) & \text{if } v=c \text{ and } u\in\{3n+1, 3n+2, \dots, 4n\}\\
        -2(2n+2) & \text{if } v=c \text{ and } u\in\{4n+1, 4n+2, \dots, 6n\}\\
        -100n^2 & \text{if } v\in \{2n+1, 2n+2, \dots, 4n\} \text{ and } u=v-2n
    \end{cases}
\end{equation*}
These constraints are illustrated in Figure~\ref{fig:long_star_vcsp}.



Consider a steepest ascent starting at \(01^{2n}0^{4n}\), 
i.e., leaf variables \(1, \dots, 2n\) start at value \(1\) and all other variables start at value \(0\).
Note that because variables \(1, \dots, 2n\) start at assignment \(1\), the overpowering negative binary constraint \(-100n^2\) ensures that a variable \(u \in \{2n+1, \dots, 4n\}\) cannot flip to \(1\), until its attached leaf variable has flipped to \(0\).

Now, since variables \(4n+1, \dots, 6n\) have the largest unary constraint at the start of the ascent, they all flip to \(1\).
After this point, variable \(1\), having induced unary constraint \(-2n-2\) will flip to \(0\), ``unlocking" variable \(2n+1\), which immediately flips to \(1\) given its induced unary of \(2n+2\).
At this point, the induced unary on the centre variable is given by
\begin{equation*}
    w(\{c\})+w(\{c,2n+1\})+\sum_{i=4n+1}^{6n}w(\{c,i\}) \ = \ 2n+2+2(2n)(2n+2)+2n(-2(2n+2)) = 2n+2.
\end{equation*}
leading to a flip of the centre to \(1\) followed by variables \(4n+1, \dots, 6n\) all flipping back to \(0\).

By inspecting the induced unaries, we see that the next variable to flip is variable \(n+1\), having induced unary constraint \(-2n-1\), inducing it to flip to \(0\).
This ``unlocks" variable \(3n+1\), which has induced unary 
\begin{equation*}
    w(\{3n+1\})+w(\{c,3n+1\}) = (2n+2)(4n+3)-2(2n+1)(2n+2) = 2n+2
\end{equation*}
leading to a flip to \(1\).
At this point, the induced unary on the centre variable is given by
\begin{equation*}
    w(\{c\})+w(\{c,2n+1\})+w(\{c,3n+1\})=2n+2+2(2n)(2n+2)-2(2n+1)(2n+2)=-(2n-2)
\end{equation*}
leading to a flip of the centre to \(0\) followed by variables \(4n+1, \dots, 6n\) all flipping back to \(1\).

At this point, variable \(2\) is the next variable to flip, followed by variable \(2n+2\), which in turn is followed by the centre variable back to \(1\), followed by flips of variables \(4n+1, \dots, 6n\).
In general, the steepest ascent follows a structure where every time a leaf variable \(i\in\{1,\dots,2n\}\) flips, this flip is followed by a flip of variable \(2n+i\), which induces a flip of the centre variable, which induces flips of variables \(4n+1, \dots, 6n\).
Thus, each flip of a variable \(i\in\{1,\dots,2n\}\), is followed by \(2n+2\) flips before the next variable in \(\{1,\dots,2n\}\) flips, meaning they are followed by a total of \(2n(2n+2)\) flips.
Since each of these variables also makes a single flip themselves, and variables \(4n+1, \dots, 6n\) also flip once each at the start of the ascent, the total number of flips is equal to \(2n(2n+2)+4n=4n^2+8n\).

\vspace{0.2cm}
 This example is very similar to Example~\ref{ex:star_vcsp_quad_ascent},
 with weights multplied by $2n + 2$ and the set of leaf variables \(\{1,\dots, 2n\}\) added, which together force the steepest ascent in this example to induce the longest ascent on the instance restricted to \(\{0\}\cup\{2n+1,\dots,6n\}\).
\end{example}

\section{Consequences for the parameterized complexity of local search}
\label{sec:PCLS}

Our star-of-gadgets construction and the prior chain-of-gadgets constructions (overviewed in \cref{sec:chaingadgets,fig:previous-constructions}) can be seen as part of a broader research programme exploring the parameterized complexity of local search.
The two keys to progress in this area are to identify suitable parameters of problem instances that can be used to 
(1) obtain limits on the complexity of local search for those instances, and
(2) distinguish between which local search algorithms will be efficient or not on those instances.
Under the second criteria, a successful parameter would delineate the complexity of instances by the kind of strict local search algorithms that succeed on them.
For example: 
for a very low value of the parameter, (a) all local search algorithms are efficient; 
for low values only (b) some local search algorithms are inefficient, but certain popular algorithms (say greedy local search) are still efficient;
for intermediate values (c) even popular locals earch algorithms become inefficient,  but there are still short ascents to peaks, so \emph{some} strict local search algorithm might still be efficient;
and finally for high values of this parameter, (d) all ascents from some initial assignment become exponentially long.
As a first test of of the success of such a hypothetical parameter, we would need to create the space to distinguish between strict local search algorithms by showing that there are values of the parameter for which there exists an exponential gap between the longest vs shortest  ascent.

Prior work focused on the max-degree of the constraint graph, and on notions such as treewidth, and the more restrictive pathwidth. 
This work has reduced both the max-degree and the pathwidth to (more or less) as low as they can go, giving us a (close to) tight characterization of how simple VCSPs can be in terms of max-degree and pathwidth while still having some ascents, the steepest ascent, or some shortest ascent, exponentially long.
However, this prior work did not identify any parameter values where there was a significant gap between the longest and shortest ascent in the worst case (see Table~\ref{tab:PC_BooleanVCSP}).

The earliest results focused on the \textit{maximum degree} of the constraint graph,
but this proved to be too coarse-grained to be very useful as a complexity parameter.
Just by going from max-degree two to three, Boolean VCSPs go from P to NP~\cite{maxcut-deg3-NP}
and longest ascents go from quadratic~\cite{KazThesis} to exponential~\cite{repCP,pw2MSc,slow-greed}.

Treewidth and pathwidth are more fine-grained parameters, and can be used to separate local search from non-local search since VCSPs with bounded treewidth are tractable for non-local dynamic programming algorithms~\cite{BB73,VCSPsurvey} but are still expressive enough to encode the exponential ascents that can make strict local search inefficient~\cite{VCSPpw3_allExp}.
However, these measures are still too coarse-grained for discriminating between different strict local search algorithms for Boolean VCSPs:
at treewidth one (i.e., trees), all ascents are at most quadratic~\cite{repCP};
by pathwidth two, the steepest ascent can be exponentially long, so popular algorithms like greedy local search can be inefficient~\cite{slow-greed,pw2MSc};
by pathwidth three, even the shortest ascent from some initial assignments 
can be exponentially long, so all strict local search algorithms can be inefficient~\cite{VCSPpw3_allExp}.
Hence, for a more fine-grained parameterized complexity analysis of local search, we are forced to look for more subtle properties of constraint graphs.

Now that we have shown that a star of gadgets can work where only a chain of gadgets was known before, 
it makes sense to move from pathwidth -- 
which measures how how far the constraint graph is from a path/chain --
to \textit{treedepth }-- which measures how far the constraint graph is from a star.
Given a tree $T=(V,P)$ rooted at $r_0$, let $\prec_T$ be the descendant relationship for $T$.
In other words, $r \prec_T s$ if the path from $r$ to $r_0$ in $T$ passes through $s$ and $r \neq s$.
We say that $T$ is an \emph{elimination-tree} of the binary VCSP-instance $\mathcal{C} = (V,\mathcal{S},w)$ if for every $\{u,v\} \subseteq S \in \mathcal{S}$ either $u \prec_T v$ or $v \prec_T u$.
For a connected $\mathcal{C}$, the $\text{treedepth}(\mathcal{C})$ is the minimal $\text{height}(T)$ over all elimination-trees of  $\mathcal{C}$.
For a disconnected $\mathcal{C}$, $\text{treedepth}(\mathcal{C})$ is the maximum over its connected components.
Treedepth is closely related to treewidth and pathwidth: 
any graph $G$ has $\text{treewidth}(G) \leq \text{pathwidth}(G) \leq \text{treedepth}(G) - 1 \leq \text{treewidth}(G)\log n$~\cite{td_old1,BGHK95,td_chapter}.

The only constraint graphs that have treedepth $1$ are VCSPs with only unary constraints.
Unary VCSPs represent fitness landscapes where each variable has a preferred assignment independent of all other variables.
Thus, there is a single fitness peak and all ascents to that peak are (very) short: all ascents have length equal to the number of variables on which the initial assignment and the peak differ.
Since an ascent can change at most one variable per step, these are the shortest possible ascents in any landscapes, so we give them the special name of \emph{ascents of Hamming-distance}.
Thus, for treedepth one: the longest ascent, steepest ascent, and shortest ascent all have the same length, so there is no separation in the linear bound.

The only constraint graphs that have treedepth $2$ are disjoint unions of stars.
In terms of bounds on all ascents, stars are trees and trees have at most quadratic ascents~\cite{repCP}; \cref{ex:star_vcsp_quad_ascent} described 
a star-structured binary Boolean VCSP that saturates this bound.
In terms of bounds on the \emph{steepest} ascent, we showed that for any fitness landscape represented by a star-structured VCSP the steepest ascent has length at most $2(n - 1)$ (\cref{prop:star_sa_is_linear}); \cref{ex:longest_star_graph_sa} saturates this bound.
So for ``some ascent'' vs ``steepest ascent'' we see a polynomial separation 
at treedepth two.
We can increase this polynomial separation slightly more for the shortest ascent, by showing that 
any constraint graph of treedepth two has a Hamming-distance ascent from every assigmnment:

\begin{lemma}
Given a binary Boolean VCSP $\mathcal{C} = (\{c\} \cup V,\mathcal{S},w)$ that is star-structured with center variable $c$ and set of leaves $V$ and any initial assignment $x_0$ to $c$ and $x^0 \in \{0,1\}^V$ to the leaves, the following are all true:
\begin{enumerate}
\item The fitness landscape of $\mathcal{C}$ has at most two peaks, and if both exist then they differ on $c$: 
i.e., they have the form $bx^{*b}$ for an assignment $b \in \{0,1\}$ to the variable $c$ and a corresponding assignment $x^{*b} \in \{0,1\}^V$ to the leaves.
\item If the peak with $b = x_0$ exists then there is an ascent to it of Hamming-distance (i.e., $||x^0, x^{*b}||_H$).
\item If the peak with $b = 1 - x_0$ exists and can be reached from $x_0 x^0$ then it can be reached by an ascent of Hamming-distance (i.e., $1 + ||x^0, x^{*b}||_H$).
\end{enumerate}
\label{thm:short_star}
\end{lemma}
\begin{proof}
Consider the two induced subproblems $\mathcal{C}^b$ on $V$ with $x[c] = b$.
Each of these subproblems has only unary constraints, 
so each induced subproblem has one fitness peak $x^{*b} \in \{0,1\}^V$.
For $x^* \in \{0,1\}^{\{c\} \cup V}$ to be a fitness peak of $\mathcal{C}$, the assignment must be a peak in $\mathcal{C}^{x^*_c}$ and the $c$ variable must not want to flip.
Thus, there are at most two peaks and they have the form $x^* = bx^{*b}$. This establishing case (1).

For notational convenience, for the rest of the proof, relabel all the domains so that $x_0 = 0$ and $x^0 = 0^{|V|}$.
Now, to get case (2), flip from $0$ to $1$ all $v \in V$ such that $w(v) > 0$.
For case (3), consider any ascent from $00^{|V|}$ that ends at $1x^{*1}$.
Look at the steps in this ascent before variable $c$ flips from $0$ to $1$.
Suppose that one of these steps flips a $v \in V$ with $w(\{c,v\}) < 0$ then we can remove this step without breaking the ascent.
Repeatedly do this.
\end{proof}

Given our \cref{ex:longest_star_graph_sa} of a quadratic steepest ascent from a starlike tree, there is no separation between longest and steepest ascents for trees.
But a small separation of quadratic versus linear is still possible for steepest versus shortest ascents from trees.
This is because a slightly less precise version of \cref{thm:short_star} (that does not characterize the exact structure of the local peaks) also holds for trees:

\begin{theorem}
   Given a binary Boolean VCSP $\mathcal{C} = (V,\mathcal{S},w)$ that is tree-structured, there exists an ascent of Hamming-distance from any initial assignment $x^0 \in \{0,1\}^V$.
    \label{thm:short_trees}
\end{theorem}

\begin{proof}
We prove this by induction on the size of the tree.
For $n = 1$, there is a single variable so it can flip at most once.
Assume that each variable flips at most once for a tree on $n$ variables and consider a tree on $n + 1$ variables.
Let $u$ be a leaf and $v$ its parent.
Relabel all domains so that $x^0 = 0^{n  + 1}$.
Consider the following four exhaustive cases of constraints on $u$:
\begin{enumerate}
\item If $w(\{u\}) < 0$ and $w(\{u\}) < w(\{u,v\})$ then $u$ will never want to flip, and the result follows from the inductive hypothesis on the induced subproblem on the rest of the tree with $x_u = 0$.
\item If $w(\{u\}) > 0$ and $w(\{u\}) > w(\{u,v\})$ then $u$ will want to flip to $x_u = 1$ regardless of the value of $x_v$.
So flip $u$ to $1$ and the result follows from the inductive hypothesis on the rest of the tree with $x_u = 1$.
\item If $0 > w(\{u\}) > -w(\{u,v\})$ then use the inductive hypothesis on the rest of the tree with $x_u = 0$ to find a local peak $x^* \in \{0,1\}^n$ while flipping each variable in the rest of the tree at most once.
If $x^*_v = 0$ then $u$ is at equilibrium and the result follows.
Otherwise, if $x^*_v = 1$ then flipping $u$ to $1$ only increases the induced unary on $v$, so keeping it at $1$, and the result follows after one flip of $u$.
\item If $0 < w(\{u\} < -w(\{u,v\})$ then use the same argument as above but $u$ flips when $x^*_v = 0$ and stays at $0$ when $x^*_v = 1$.
\end{enumerate}
\end{proof}

Unfortunately, \cref{thm:short_trees} does not translate to a more general statement about treedepth since trees can have treedepth of upto $\log n + 1$. 
But all the prior chain of gadgets constructions with exponential longest, steepest, and even shortest ascents that we described in \cref{sec:chaingadgets} also have treedepth of $\Theta(\log n)$. 
This is because for any VCSP on $n$ variables constructed as a chain of $m$ gadgets with each of constant size $\leq k$: 
(1) will include a path in the constraint graph of length linear in the number of variables and hence have treedepth of $\Omega(\log n)$~\cite{BGHK95,td_chapter}; and 
(2) have a treedepth of $O(k \log m)$ by treating each gadget as a single node, building an elimination tree of depth $\log m + 1$ for the resulting path of $m$ nodes and then replacing each node by the $k$ variables in the gadget.
In contrast, our star of gadgets construction in \cref{fig:triangle-flower} has treedepth three:
\begin{proposition}
The construction $\C_{\leq n}$ shown in \cref{fig:triangle-flower} has treedepth three.
\label{prop:td3}
\end{proposition}
\begin{proof}
    Since the complete graph on three vertices is a subgraph of $\C_{\leq n}$, we have $\text{td}(\C_{\leq n})\geq \text{td}(K_3)= 3$.
    To establish that $\text{td}(\C_{\leq n})\leq 3$, consider the elimination tree of $\C_{\leq n}$ with variable $0$ as the root and edges $\{0,(1,k)\},\{(1,k),(2,k)\},\{0,(3,k)\},\{(3,k),(4,k)\}$ for all $k\in[n]$, which has depth $3$.
\end{proof}
Hence we have shown that by moving from treedepth two to treedepth three in the constraint graph we move from all ascents being 
necessarily at most quadratic in length to the possibility of encoding a landscape with an exponential ascent. 
This raises the natural questions: is there a similar 
threshold in terms of treedepth for the possibility of encoding an exponential \textit{steepest} ascent, or for 
encoding an exponential \textit{shortest} ascent? 

As a first step, we ask whether it is possible to construct a constraint graph 
of treedepth three which can encode a landscape with an exponential \textit{steepest} ascent.
Note that our star of gadgets construction (\cref{ex:star_of_gadgets}) has an exponential ascent from the initial assignment $0^n$, 
but the \textit{steepest} ascent from this initial assignment 
is a single step flipping the central variable to immediately reach the peak $1(0)^{n - 1}$.
The starlike tree in \cref{ex:long_star} has a treedepth of $3$ and represents a fitness landscape with a quadratic steepest ascent.
This example was developed using established tricks for converting a long ascent into a steepest ascent~\cite{pw4,pw2MSc}, 
but these tricks increase the treedepth of the constraint graph by at least one. 
We conjecture that a constraint graph of treedepth at least four will be required to be able to encode an exponential steepest ascent.

However, for \textit{shortest} vs longest ascents we can prove that an exponential gap exists at treedepth three.
Specifically, we now show that the length of the \textit{shortest} ascent from any initial assignment can be at most linear:
\begin{theorem}
    Given a Boolean VCSP $\mathcal{C} = (V,\mathcal{S},w)$ of treedepth three,
    from any initial assignment $x_0 \in \{0,1\}^V$ there exists an ascent of length $\leq 3|V| - 1$.
    \label{thm:short_td3}
\end{theorem}

\begin{proof}
For simplicity of notation, relabel our variables so that their names follow the elimination-tree $T$ of the VCSP:
\begin{equation}
V = \{r\} \cup \Big( \bigcup_{i = 1}^m \{c_i\} \cup V_i \Big)
\end{equation}
where $r$ is the root of $T$ and $V_i$ are the leaves of $T$ at depth three whose parent is the variable $c_i$ at depth two.
For elimination-forests, analyse each elimination-tree separately.

Define the induced subproblems $\mathcal{C}^b_{i}$ on $\{c_i\} \cup V_i$ with background $x_r = b$.\footnote{
By definition of elimination-tree, there are no constraints between variables $\{c_i\} \cup V_i$ and $\{c_j\} \cup V_j$ for $i \neq j$.
Hence, we do not need to specify the background assignment to $c_j$ or the variables in $V_j$, only an assignment to $x_r$.}
The constraint graphs of these induced subproblems are stars, so \cref{thm:short_star} will apply to them.

Now consider the following ascent-following algorithm that works in five phases:
\begin{enumerate}
\item Bring each $\mathcal{C}_i^{x^0_r}$ to its local peak quickly (i.e., following \cref{thm:short_star}) while keeping the $r$ variable fixed.
If the initial assignment allows two different peaks to be reached in $\mathcal{C}_i^{x_r}$ then follow the short ascent to the higher fitness peak.
For simplicity of notation, relabel the domains so that the assignment at the end of this step is $x^{\scriptsize \circled{1}} = 0^{|V|}$.
\item If flipping $x^{\scriptsize \circled{1}}_r$ from $0$ to $1$ increases fitness then flip it; otherwise we have found a fitness peak.
\item Bring each $\mathcal{C}^1_i$ to its local peak quickly.
If $x^{\scriptsize \circled{1}}$ allows two different peaks to be reached in $\mathcal{C}_i^{x_r}$ then follow the short ascent to the peak with $x[c_i] = 0$.
Call resulting assignment $x^{\scriptsize \circled{3}}$.
\item If flipping $x^{\scriptsize \circled{3}}_r$ from $1$ to $0$ increases fitness then flip it; otherwise we have found a fitness peak.
\item Bring each $\mathcal{C}^0_i$ to its local peak quickly.
If $x^{\scriptsize \circled{3}}$ allows two different peaks to be reached in $\mathcal{C}_i^{x_r}$ then follow the short ascent to the peak with $x[c_i] = 0$.
Call resulting assignment $x^{\scriptsize \circled{5}}$.
\end{enumerate}
Note that in this algorithm, if the fitness landscape corresponding to $\mathcal{C}^1_i$ has two fitness peaks then, by \cref{thm:short_star}, we will have $x[c_i] = 0$ at all steps of the algorithm that are after the notational relabelling in phase (1).

Now, we will show that $x^{\scriptsize \circled{5}}$ is a local peak.
We proceed by contradiction: suppose that in $x^{\scriptsize \circled{5}}$, variable $r$ wants to flip to $1$.
In that case -- as a new phase (6) -- flip it to $1$, and then -- as a new phase (7) -- repeat phase (3) but call the resulting assignment $x^{\scriptsize \circled{7}}$.
Since we were following an ascent, we must have $\mathcal{C}(x^{\scriptsize \circled{7}}) > \mathcal{C}(x^{\scriptsize \circled{3}})$.
That means we must have at least one $i$ such that the peak in $\mathcal{C}^1_{i}$ is different between $x^{\scriptsize \circled{3}}$ and $x^{\scriptsize \circled{7}}$.
But if there were two fitness peaks then by the observation in the previous paragraph we would have $x^{\scriptsize \circled{3}}[c_i] = x^{\scriptsize \circled{7}}[c_i] = 0$ and by \cref{thm:short_star} they would be the same peak.
A contradiction.

From this, we can see that $x^{\scriptsize \circled{5}}$ is a local peak and thus there is an ascent in $\mathcal{C}$ from any initial assignment to the local peak that flips $r$ only twice.
Each variable in $\bigcup_{i = 1}^m (\{c_i\} \cup V_i)$ flips at most three times (at most once in each of phase 1, 3, and 5).
Thus the shortest ascent is $\leq 3|V| - 1$.
%
\end{proof}

As far as we know, \cref{thm:short_td3} provides the first provable separation between some ascent exponential versus shortest ascent polynomial for a fixed structural parameter of a constraint graph.\footnote{
If we assign arc-directions to the scopes of binary Boolean VCSPs, as done by \textcite{effective-and-efficient}, then a prior separation exists for the structural parameter of being a directed acyclic graph.
Oriented VCSPs are the only VCSPs with DAG-structured constraint graphs. Oriented VCSPs always have Hamming-distance ascents 
and even stronger: \textcite{effective-and-efficient} proved that reasonable local search algorithms like random-ascent find short ascents with high probability.
However, both the \textcite{hakenSteepest} and \textcite{slow-greed} constructions of exponential steepest ascent are oriented VCSPs.
Thus, there is an exponential separation between steepest ascent and random ascent for the `DAG' structural parameter.
}
And we believe our results can be pushed further because all previous constructions where the shortest ascent is exponential have unbounded treedepth (at least logarithmic in the number of variables) and we see no way to avoid this.
However, to extend our results to higher treedepth seems to need proof techniques that go beyond analyzing the structure of local peaks.
So at this time, we can only conjecture that unbounded treewidth is a necessary condition for exponential shortest ascents, or equivalently:
\begin{conjecture}
If a family of VCSPs has its treedepth bounded by a constant then given any initial assignment, there exists an ascent of polynomial length to some local peak.
\label{conj:treedepth}
\end{conjecture}

Our intuition for the unavoidability of unbounded treedepth to achieve exponential 
shortest ascents comes from the structure of the chain constructions that we discussed in \ref{sec:chaingadgets}.
All previous chain of gadget constructions where even the shortest ascents are 
exponentially long have the following property: at any given step only one gadget is ``active'', where  a gadget is ``active'' if it has some potential change in its variables that is fitness-increasing.
This ``activation'' of the gadget is moved left or right along the chain as the ascent proceeds.
For a ``star'' of gadgets, there is no similar concept of left or right, 
so to force the shortest ascents to be exponential we
expect to have to expand the ``centre'' of the ``star'' to  achieve a similar property of specifying which gadget is ``activated''.
A ``centre'' with this property would require on the order of $\log(n)$ variables. 

Building on this intuition, 
we now obtain a polynomial bound on the length of shortest ascents by using the 
parameter of \emph{vertex cover number} instead of treedepth.
We say that a VCSP $\mathcal{C} = (V,\mathcal{S},w)$ on $n$ variables has vertex cover number $k$ if there exist $k$ variables $K \subseteq V$ such that the induced subproblem of $\mathcal{C}$ on $V \setminus K$ in any background $z \in \{0,1\}^K$ has at most unary constraints.
Vertex cover number is a more fine-grained parameter than treewidth, pathwidth, and treedepth: 
$\text{treewidth}(G) \leq \text{pathwidth}(G) \leq \text{treedepth}(G) - 1 \leq \text{vertexcover\#}(G)$. 
A star-structured VCSP has vertex cover number $1$.
We can think of a VCSP with vertex cover number $k$ as a starlike graph that allows the centre to consist of $k$ variables while still restricting the other $n - k$ `leaf' variables to form an independent set.
This allows us to formalize our above intuition about the necessary conditions for exponential shortest ascents by using vertex cover number to measure the size of the ``large centre''.
Specifically, if this vertex cover number is $O(\log n)$ then we prove that the resulting fitness landscapes have a polynomial shortest ascent from any initial assignment:

\begin{theorem}
If a Boolean VCSP-instance $\mathcal{C}$ has vertex cover number $k$ then the shortest ascent is $\leq 2^k(n - k + 1) - 1$.
\label{thm:vcn_shortest_upperbound}
\end{theorem}

\begin{proof}
Let $K \subseteq V$ be a vertex cover of size $k$.
For each $z \in \{0,1\}^K$, define $\mathcal{C}^z$ as the induced subproblem on $V \setminus K$ with background $z$.
Since $K$ is a vertex cover, it means there are only unary constraints in each $\mathcal{C}^z$, so each has a unique fitness peak $y^{*z}$ that can be reached by an ascent of Hamming-distance.
Now consider the following two-phase ascent-following algorithm:

\begin{enumerate}
\item Given a current assignment $x = zy$ with $z \in \{0,1\}^K$, bring $\mathcal{C}^z$ to its local peak $y^{*z}$ quickly.
\item If there is some flip in $K$ that increases fitness, make it and return to phase (1).
\end{enumerate}

\noindent Notice that at the end of phase (1), the assignment has the form $zy^{*z}$. 
Since each $z$ has a unique $y^{*z}$, there are only $2^k$ such assignments.
An ascent cannot revisit the same assignment.
At most $n - k$ steps are taken in phase (1) and one step in phase (2), thus the overall ascent has length $\leq 2^k(n - k + 1) - 1$.
\end{proof}

Unfortunately, our star of gadgets on $4m + 1$ variables in \cref{ex:star_of_gadgets} has vertex cover number $2m + 1$ and so unlike the treedepth case, it cannot be combined with \cref{thm:vcn_shortest_upperbound} to show a separation between longest and shortest ascents.
Such a separation is possible with a more complicated construction of a starlike graph with a logarithmic number of variables in the centre. 
The construction is easiest to describe if we allow higher-arity constraints on the centre 
variables and then specify the induced unary constraints that these impose on the leaves.
\begin{example}\label{ex:vcn_lowerbound}
\textbf{(Starlike-VCSP with exponential steepest ascent)}
Consider a Boolean VCSP $\mathcal{C}$ on $n := m + \log m + 2$ variables\footnote{
for simplicity, let $m$ be an integer power of $2$.} with a set of `centre' variables $\{+,-\}\cup G$ with $|G| = \log m$ and `leaf' variables $L = \{0,\ldots,m - 1\}$.
We will use the leaf variables $L$ as the `tape' for a standard binary counter.
Let $\text{gc}: \{0,1\}^G \rightarrow L$ be a Gray code -- an ordering of the Boolean numerals such that two successive values differ in only one bit.
We will use $G$ as a `tape-index' for keeping track of the `head' that points at cells on the `tape' of the standard binary counter.
Given the $4m$ background assignments $x_+x_-\;y \in \{0,1\}^{\{+,-\}}\times\{0,1\}^G$, we will specify $\mathcal{C}$ by its induced unary VCSP subproblems on the leaves $L$:
\begin{align}
\mathcal{C}_k & =  (L,U,w_k) & \text{ in background assignment } & x_+ = x_- \text{ and } y = \text{gc}^{-1}(k), & \\
\mathcal{C}^+_k & = (L,U,w^+_k) & \text{ in background assignment } & x_+ = 1, x_- = 0 \text{ and } y = \text{gc}^{-1}(k), \text{ and} &  \\
\mathcal{C}^-_k & = (L,U,w^-_k) & \text{ in background assignment } & x_+ = 0, x_- = 1 \text{ and } y = \text{gc}^{-1}(k) &
\end{align}
where $U = \{\emptyset \} \cup \{ \{i\} \; | \; i \in L\}$ is the set of nullary and unary constraints on $L$ and
\begin{align}
w^+_k(\;\emptyset\;) & = k + 1 \label{eq:vcn_lb_wk_first}\\
w^-_k(\;\emptyset\;) & = -(k + 1) \\
w_k(\;\emptyset\;) & =  \begin{cases}
  0 & \text{if } k = 0 \\
  (3m + 4) \cdot 2^k - (m + 2) & \text{if } k > 0
\end{cases} \\
w^{\pm}_k(\{i\}) & = (3m + 4) \cdot 2^i \\
w_k(\{i\}) & = \begin{cases}
-1 & \text{if } i < k \\
1 & \text{if } i = k \\
(3m + 4) \cdot 2^i & \text{if } i > k \label{eq:vcn_lb_wk_last}
\end{cases}.
\end{align}
The above can also be specified directly for $\mathcal{C}$ as $m$ general valued constraints, each of $(\log m + 3)$-arity with scopes involving the center variables $\{+,-\} \cup G$ and one of the leaf variables $L$.

We will use an ascent on this VCSP-instance starting at $00\;0^{\log m}\;0^m \in \{0,1\}^{\{+,-\}}\times \{0,1\}^G \times \{0,1\}^L$ and ending at a local peak $00\;\text{gc}^{-1}(m - 1)\;0^m$ to implement a standard binary counter on the leaves $L$.
We call this the $L$-counter.
Whenever we need to carry bits in the $L$-counter to advance the counter from $1^{k - 1}0z \in \{0,1\}^L$ to $0^{k - 1}1z$ for any $z \in \{0,1\}^{m - k - 1}$, we will have our ascents assignment at $x^{\scriptsize \circled{0}} = 00\;0^{\log m} 1^{k - 1}0z$ with fitness
\begin{align}
\mathcal{C}(x^{\scriptsize \circled{0}}) & = w_0(\emptyset) + \sum_{i = 0}^{k-1} w_0(\{i\}) + \sum_{i = 1}^{m - k - 1}w_0(\{k + i\})z_{i}\\
& = (3m + 4)\cdot(2^{k} - 1 + 2^{k+1} \cdot [z]_\text{bin} )
\end{align}
where $[z]_\text{bin}$ is the integer encoded by $z$ in binary and then follow these eight phases:
\begin{enumerate}
\item Flip $x_+$ to $1$ to allow us to move up the `tape-index' Gray code.
This gets us to $x^{\scriptsize \circled{1}} = 10\;0^{\log m} 1^{k - 1}0z$ with fitness 
\begin{align}
\mathcal{C}(x^{\scriptsize \circled{1}}) & = w^+_0(\emptyset) + \sum_{i = 0}^{k-1} w^+_0(\{i\}) + \sum_{i = 1}^{m - k - 1}w^+_0(\{k + i\})z_{i}\\
& = 1 + \mathcal{C}(x^{\scriptsize \circled{0}})
\end{align}
\item Increment the `tape-index' $i \in L$ from $0$ to $k$ by incrementing the corresponding Gray code on $G$ from $\text{gc}^{-1}(0) = 0^{\log m} \in \{0,1\}^G$ to $\text{gc}^{-1}(k) \in \{0,1\}^G$.
This gets us to $x^{\scriptsize \circled{2}} = 10\;\text{gc}^{-1}(k) 1^{k - 1}0z$ with fitness
\begin{align}
\mathcal{C}(x^{\scriptsize \circled{2}}) & = w^+_k(\emptyset) + \sum_{i = 0}^{k-1} w^+_k(\{i\}) + \sum_{i = 1}^{m - k - 1}w^+_k(\{k + i\})z_{i}\\
& = k + 1 + \mathcal{C}(x^{\scriptsize \circled{0}})
\end{align}
\item Flip $x_+$ to $0$ to allow us to enter the `clear and carry the ones' stage.
This gets us to $x^{\scriptsize \circled{3}} = 00\;\text{gc}^{-1}(k) 1^{k - 1}0z$
with fitness
\begin{align}
\mathcal{C}(x^{\scriptsize \circled{3}}) & = w_k(\emptyset) + \sum_{i = 0}^{k-1} w_k(\{i\}) + \sum_{i = 1}^{m - k - 1}w_k(\{k + i\})z_{i}\\
& = (3m + 4) \cdot 2^k - (m + 2) - (k - 1)  + (3m + 4)\cdot 2^{k+1}\cdot[z]_\text{bin} \\
& =  2m - k + 1 + (3m + 4)\cdot(2^{k} - 1 + 2^{k+1} \cdot [z]_\text{bin} ) \\
& = 2m - k + 1 + \mathcal{C}(x^{\scriptsize \circled{0}}) \\
& > \mathcal{C}(x^{\scriptsize \circled{2}})
\end{align}
where the last strict inequality follows from $k \leq m - 1$.
\item For all $i < k$, flip every $x_i$ for $i \in L$ from $1$ to $0$, thus taking the counter from $1^{k - 1}0z \in \{0,1\}^L$ to $0^{k-1}0z$.
This gets us to $x^{\scriptsize \circled{4}} = 00\;\text{gc}^{-1}(k) 0^{k - 1}0z$
with fitness
\begin{align}
\mathcal{C}(x^{\scriptsize \circled{4}}) & = w_k(\emptyset) + \sum_{i = 1}^{m - k - 1}w_k(\{k + i\})z_{i}\\
& = (3m + 4) \cdot 2^k - (m + 2) + (3m + 4)\cdot 2^{k+1}\cdot[z]_\text{bin} \\
& = \mathcal{C}(x^{\scriptsize \circled{3}}) + k - 1
\end{align}
\item Flip $x_k$ 
from $0$ to $1$, thus taking the counter from $0^{k-1}0z$ to $0^{k-1}1z \in L$.
This gets us to $x^{\scriptsize \circled{5}} = 00\;\text{gc}^{-1}(k) 0^{k - 1}1z$
with fitness
\begin{align}
\mathcal{C}(x^{\scriptsize \circled{5}}) & = w_k(\emptyset) + w_k(\{k\}) + \sum_{i = 1}^{m - k - 1}w_k(\{k + i\})z_{i}\\
& = \mathcal{C}(x^{\scriptsize \circled{3}}) + k
\end{align}
\item Flip $x_-$ to $1$ to allow us to move down the `tape-index' Gray code.
This gets us to $x^{\scriptsize \circled{6}} = 01\;\text{gc}^{-1}(k) 0^{k - 1}1z$
with fitness
\begin{align}
\mathcal{C}(x^{\scriptsize \circled{6}}) & = w^-_k(\emptyset) + w^-_k(\{k\}) + \sum_{i = 1}^{m - k - 1}w^-_k(\{k + i\})z_{i}\\
& = -(k + 1) + (3m + 4)\cdot2^k + (3m + 4)\cdot 2^{k+1}\cdot[z]_\text{bin} \\
& =  m + 2 + 2m - k + 1 + (3m + 4)\cdot(2^k - 1 + 2^{k+1}\cdot[z]_\text{bin}) \\
& = m + 2 + \mathcal{C}(x^{\scriptsize \circled{3}}) \\
& > \mathcal{C}(x^{\scriptsize \circled{5}})
\end{align}
\item Decrement the `tape-index' $i \in L$ from $k$ to $0$ by decrementing the corresponding Gray code on $G$ from $\text{gc}^{-1}(k) \in \{0,1\}^G$ to $\text{gc}^{-1}(0) = 0^{\log m} \in \{0,1\}^G$.
This gets us to $x^{\scriptsize \circled{7}} = 01\;0^{\log m} 0^{k - 1}1z$
with fitness
\begin{align}
\mathcal{C}(x^{\scriptsize \circled{7}}) & = w^-_0(\emptyset) + w^-_0(\{k\}) + \sum_{i = 1}^{m - k - 1}w^-_0(\{k + i\})z_{i}\\
& = m + k + 1 + \mathcal{C}(x^{\scriptsize \circled{3}}) \\
& = 3m + 2 + \mathcal{C}(x^{\scriptsize \circled{0}})
\end{align}
\item Flip $x_-$ to $0$ to be ready for the next increment of the $L$-counter.
This gets us to $x^{\scriptsize \circled{8}} = 00\;0^{\log m} 0^{k - 1}1z$ with fitness:
\begin{align}
\mathcal{C}(x^{\scriptsize \circled{8}}) & = w_0(\emptyset) + w_0(\{k\}) + \sum_{i = 1}^{m - k - 1}w_0(\{k + i\})z_{i} \\
& = (3m + 4)\cdot2^k + (3m + 4)\cdot 2^{k+1}\cdot[z]_\text{bin} \\
& = 3m + 4 + (3m + 4)\cdot(2^k - 1 + 2^{k+1}\cdot[z]_\text{bin}) \\
& = 3m + 4 + \mathcal{C}(x^{\scriptsize \circled{0}}) \\
& > \mathcal{C}(x^{\scriptsize \circled{7}}) 
\end{align}
\end{enumerate}
Via our calculations of fitness above, we check that the induced unaries specified in \cref{eq:vcn_lb_wk_first}-\eqref{eq:vcn_lb_wk_last} will allow the $L$-counter to execute the eight phases via ascent steps for all binary integers until $1^m$.
In the case of $k = 0$, we do not need to worry about carry bits, and so can execute just phase 5 instead of all eight phases.

Since our ascent passes through each assignment $z \in \{0,1\}^L$, we know that it's length is $\geq 2^m$ which, given that the number of variables $n = m + \log m + 2$, is $2^{\Omega(n)}$.
Finally, given that $\{+,-\} \cup G$ is (by design) a vertex cover of $\mathcal{G}$ and $|\{+,-\} \cup G| = \log m + 2 \in O(\log n)$, we have an 
exponential longest ascent in a VCSP-instance with a logarithmic vertex cover.
\end{example}

In other words, \cref{ex:vcn_lowerbound} as a lower bound on the longest ascent together with \cref{thm:vcn_shortest_upperbound} for the upper bound on shortest ascent from any initial assignment provide our second provable separation between some ascent exponential versus shortest ascent polynomial for a structural parameter of a constraint graph.
Thus, both treedepth three and logarithmic vertex cover number can provide the space needed to distinguish between different strict local search algorithms.
It is just that in the case of \cref{ex:star_of_gadgets} for treedepth, introducing just one additional variable with binary constraints allowed us to intertwine two simple landscapes into a complex one with long ascents.
But in the case of \cref{ex:vcn_lowerbound} for vertex cover number,
introducing $\log m + 2$ additional variables with $(\log m + 3)$-arity  constraints allowed us to intertwine $4m$ of the simplest possible landscapes into a complex one with long ascents. 
We expect it is possible to replace these $(\log m + 3)$-arity constraints by binary ones without significantly increasing the vertex cover number.

\section{Conclusion}

To understand when local search is efficient and when it is not, 
we studied which fitness landscapes are easy to navigate and which are hard.
We showed that adding just one additional dimension that connects two simple and easy to navigate sublandscapes where all ascents are linear can create a complex landscape that harbours exponentially long ascents.
In terms of the valued constraint satisfaction problems that represent these landscapes, this required us to move from the well-established chain of gadgets constructions to a new kind of construction based on a star of gadgets.
The single variable at the centre of this star of gadgets provided the extra dimension, and the binary constraints on this additional variable intertwined two otherwise simple VCSP subproblems into something much more complex.

Our shift to considering stars of gadgets pushed us to consider the treedepth and the vertex cover number of the associated constraint graphs as new parameters for the complexity of local search.
As a result, we were able to establish treedepth three and logarithmic vertex cover number as the first graph parameters where we can show that the shortest ascent from any initial assignment is polynomial while the longest ascent can be exponential.
This gap makes the question of the tractability of finding local peaks into a question of whether specific local search algorithms end up following one of the short ascents or one of the long ones.
We started to study this question for the steepest ascents, which are the ones followed by greedy local search.
It would be interesting for future work to consider other popular algorithms that are used in engineering practice (random ascent~\cite{randomFitter1}, simulated annealing, or advanced simplex-pivot-rules like random facet~\cite{randomFacetBound,effective-and-efficient}) or even forced on us by nature when we are modelling energy minimization in physical systems, biological evolution~\cite{W32, evoPLS, KazThesis} or social dynamics in business~\cite{orgBehavOrig,orgBehavNetwork} and economics~\cite{R09}.
Overall, our results suggest that studying the parameterized complexity of valued constraint satisfaction problems can both help us to better understand when local search is efficient and teach us something deeper about how complexity emerges in high-dimensional fitness landscapes.


\printbibliography


\end{document}